\documentclass[runningheads]{llncs}
\newif\iflong
\IfFileExists{tmp-responsibility-analysis-short-long.tex}{\longtrue
}{\longfalse}
\iflong
\else
\fi
\usepackage{CJKutf8} 

\usepackage{booktabs}   
\usepackage{subcaption} 
\usepackage{enumitem}   
                        
\usepackage{amssymb}
\usepackage{amsmath}
\usepackage{xcolor}
\definecolor{darkblue}{rgb}{0,0,0.35}
\usepackage{hyperref}
\definecolor{DarkBlue}{cmyk}{0.99998,1,0,0}
\hypersetup{breaklinks=true,pageanchor=true,hyperindex=true,hidelinks,bookmarks=false,colorlinks=true,}
\AtBeginDocument{\hypersetup{allcolors=DarkBlue}}

\usepackage{eqntabular,renumber}
\usepackage{listings}
\usepackage{galois}
\usepackage{calculus}
\usepackage[ntheorem]{mdframed}

\usepackage{wrapfig}

\newcommand{\naturalnum}{\mathbb{N}}
\newcommand{\composition}{\circ}
\newcommand{\sqb}[1]{\llbracket#1\rrbracket}
\usepackage{tikz}
\usetikzlibrary{calc}
\usetikzlibrary{decorations}
\usepackage{adjustbox}
\usetikzlibrary{arrows}
\usetikzlibrary{snakes}
\usetikzlibrary{positioning}
\usetikzlibrary{automata}
\usepackage{interval}
\usepackage{relsize}
\makeatletter
\newbox\mystrutbox
\newcommand{\@ustrut}[1]{\setbox\mystrutbox\hbox{#1\strut}\hbox{\vrule 
     height\ht\mystrutbox
     depth\z@
     width\z@}}
\newcommand{\@lstrut}[1]{\setbox\mystrutbox\hbox{#1\strut}\hbox{\vrule 
     height\z@
     depth\dp\mystrutbox
     width\z@}}
\newcommand{\ustrut}{\@ustrut{}}
\newcommand{\lstrut}{\@lstrut{}}
\newcommand{\ulstrut}{\@ustrut{\relsize{1}}}
\newcommand{\uLstrut}{\@ustrut{\relsize{2}}}
\newcommand{\uhstrut}{\@ustrut{\relsize{3}}}
\newcommand{\uHstrut}{\@ustrut{\relsize{4}}}
\newcommand{\llstrut}{\@lstrut{\relsize{1}}}
\newcommand{\lLstrut}{\@lstrut{\relsize{2}}}
\newcommand{\lhstrut}{\@lstrut{\relsize{3}}}
\newcommand{\lHstrut}{\@lstrut{\relsize{4}}}

\newcommand{\usstrut}{\@ustrut{\relsize{-1}}}
\newcommand{\lsstrut}{\@lstrut{\relsize{-1}}}
\newcommand{\ussstrut}{\@ustrut{\relsize{-2}}}
\newcommand{\lssstrut}{\@lstrut{\relsize{-2}}}
\newcommand{\usssstrut}{\@ustrut{\relsize{-3}}}
\newcommand{\lsssstrut}{\@lstrut{\relsize{-3}}}
\makeatother
\newcommand{\ltuple}[1]{\langle#1,\allowbreak}
\newcommand{\mtuple}[1]{\:#1,\allowbreak}
\newcommand{\rtuple}[1]{\:#1\rangle}

\newcommand{\pair}[2]{\ltuple{#1}\rtuple{#2}}
\newcommand{\triple}[3]{\ltuple{#1}\mtuple{#2}\rtuple{#3}}
\newcommand{\quadruple}[4]{\ltuple{#1}\mtuple{#2}\mtuple{#3}\rtuple{#4}}

\newcommand{\sextuple}[6]{\ltuple{#1}\mtuple{#2}\mtuple{#3}\mtuple{#4}\mtuple{#5}\rtuple{#6}}

\newcommand{\event}{\textrm{\itshape{e}}}
\newcommand{\setofallevents}{\textrm{\upshape E}}
\newcommand{\trace}{\sigma}
\newcommand{\setofallnonemptytrace}{\setofallevents^{+\infty}}%
\newcommand{\emptytrace}{\varepsilon}
\newcommand{\setofalltraces}{\setofallevents^{\ast\infty}}%
\newcommand{\isprefix}{\preceq}

\makeatletter
\newcommand{\prefixesofset}[1]{\def\@first{#1}\mathsf{Pref}\ifx\@first\@empty\else(#1)\fi}
\makeatother

\newcommand{\gap}{\triangleright}

\makeatletter
\newcommand{\semantics}{\mathcal{S}}%

\newcommand{\maximalsemantics}[1]{\def\@first{#1}\mathcal{S}^{\mathsf{Max}}\ifx\@first\@empty\else\sqb{#1}\fi}%
\newcommand{\prefixsemantics}[1]{\def\@first{#1}\mathcal{S}^{\mathsf{Pref}}\ifx\@first\@empty\else\sqb{#1}\fi}%
\makeatother

\newcommand{\traceproperty}{\mathcal{P}}

\newcommand{\tracepropertyT}{\mathcal{T}}
\newcommand{\predictiontraceproperty}{\mathcal{Q}}

\makeatletter
\newcommand{\alphapropertytransform}[2]{\def\@first{#1}\def\@second{#2}\alpha_{\mathsf{Pred}}\ifx\@first\@empty\else\sqb{#1}\fi\ifx\@second\@empty\else({#2})\fi}%
\newcommand{\gammapropertytransform}[2]{\def\@first{#1}\def\@second{#2}\gamma_{\mathsf{Pred}}\ifx\@first\@empty\else\sqb{#1}\fi\ifx\@second\@empty\else({#2})\fi}%
\newcommand{\alphapropertysettransform}[2]{\def\@first{#1}\def\@second{#2}\bar{\alpha}_{\mathsf{Pred}}\ifx\@first\@empty\else\sqb{#1}\fi\ifx\@second\@empty\else({#2})\fi}%
\makeatother

\newcommand{\latticeofproperties}{\mathcal{L}}%

\newcommand{\latticeofmaximalproperties}{\mathcal{L}^{\mathsf{Max}}}%
\newcommand{\maximaltracepropertytop}{\top^{\mathsf{Max}}}%
\newcommand{\maximaltracepropertybot}{\bot^{\mathsf{Max}}}%
\newcommand{\maximaltracepropertylessthan}{\subseteq}%
\newcommand{\maximaltracepropertyjoin}{\mathop{\rlap{\Large$\mskip1.7mu\cdot$}\cup}}%
\newcommand{\bigmaximaltracepropertyjoin}{\mathop{\rlap{\raisebox{-0.3ex}[0pt][0pt]{\LARGE$\mskip3.5mu\cdot$}}\bigcup}\limits}%
\newcommand{\maximaltracepropertymeet}{\mathop{\rlap{\Large$\mskip1.6mu\cdot$}\cap}}%
\makeatletter
\newcommand{\observation}[3]{\def\@first{#1}\def\@second{#2}\def\@third{#3}\mathbb{O}\ifx\@first\@empty\else\{#1\}\fi\ifx\@second\@empty\else\sqb{#2}\fi\ifx\@third\@empty\else({#3})\fi}%
\makeatother

\makeatletter
\newcommand{\cognizancefunction}[2]{\def\@first{#1}\def\@second{#2}\mathbb{C}\ifx\@first\@empty\else\sqb{#1}\fi\ifx\@second\@empty\else({#2})\fi}%
\newcommand{\omniscientcognizance}[2]{\def\@first{#1}\def\@second{#2}\mathbb{C}_{o}\ifx\@first\@empty\else\sqb{#1}\fi\ifx\@second\@empty\else({#2})\fi}%
\newcommand{\inquiryfunction}[2]{\def\@first{#1}\def\@second{#2}\mathbb{I}\ifx\@first\@empty\else\sqb{#1}\fi\ifx\@second\@empty\else({#2})\fi}%
\newcommand{\observationforsingletrace}[3]{\def\@first{#1}\def\@second{#2}\def\@third{#3}\mathbb{O}\ifx\@first\@empty\else\{#1\}\fi\ifx\@second\@empty\else\sqb{#2}\fi\ifx\@third\@empty\else({#3})\fi}%
\makeatother

\newcommand{\behavior}{\mathcal{B}}
\newcommand{\entityresp}{E_{\textup R}}

\makeatletter
\newcommand{\alpharesponsibility}[3]{\def\@first{#1}\def\@second{#2}\def\@third{#3}\alpha_{R}\ifx\@first\@empty\else\big\{#1\big\}\fi\ifx\@second\@empty\else\sqb{#2}\fi\ifx\@third\@empty\else({#3})\fi}%
\makeatother

\newcommand{\tracehistory}{\sigma_{\textup H}}%
\newcommand{\traceresponsible}{\sigma_{\textup R}}%
\newcommand{\tracefuture}{\sigma_{\textup F}}%
\makeatletter

\newcommand{\alpharesponsibilityF}[3]{\def\@first{#1}\def\@second{#2}\def\@third{#3}\alpha_{R}^{F}\ifx\@first\@empty\else\big\{#1\big\}\fi\ifx\@second\@empty\else\sqb{#2}\fi\ifx\@third\@empty\else({#3})\fi}%

\newcommand{\alpharesponsibilityH}[3]{\def\@first{#1}\def\@second{#2}\def\@third{#3}\alpha_{R}^{H}\ifx\@first\@empty\else\big\{#1\big\}\fi\ifx\@second\@empty\else\sqb{#2}\fi\ifx\@third\@empty\else({#3})\fi}%

\newcommand{\alpharesponsibilityTop}[3]{\def\@first{#1}\def\@second{#2}\def\@third{#3}\alpha_{R}^{\top}\ifx\@first\@empty\else\big\{#1\big\}\fi\ifx\@second\@empty\else\sqb{#2}\fi\ifx\@third\@empty\else({#3})\fi}%

\newcommand{\alpharesponsibilityBot}[3]{\def\@first{#1}\def\@second{#2}\def\@third{#3}\alpha_{R}^{\bot}\ifx\@first\@empty\else\big\{#1\big\}\fi\ifx\@second\@empty\else\sqb{#2}\fi\ifx\@third\@empty\else({#3})\fi}%

\newcommand{\alpharesponsibilityC}[3]{\def\@first{#1}\def\@second{#2}\def\@third{#3}\alpha_{R_{c}}\ifx\@first\@empty\else\big\{#1\big\}\fi\ifx\@second\@empty\else\sqb{#2}\fi\ifx\@third\@empty\else({#3})\fi}%

\newcommand{\alpharesponsibilityCF}[3]{\def\@first{#1}\def\@second{#2}\def\@third{#3}\alpha_{R_{c}}^{F}\ifx\@first\@empty\else\big\{#1\big\}\fi\ifx\@second\@empty\else\sqb{#2}\fi\ifx\@third\@empty\else({#3})\fi}%

\newcommand{\alpharesponsibilityCH}[3]{\def\@first{#1}\def\@second{#2}\def\@third{#3}\alpha_{R_{c}}^{H}\ifx\@first\@empty\else\big\{#1\big\}\fi\ifx\@second\@empty\else\sqb{#2}\fi\ifx\@third\@empty\else({#3})\fi}%

\newcommand{\alpharesponsibilityCTop}[3]{\def\@first{#1}\def\@second{#2}\def\@third{#3}\alpha_{R_{c}}^{\top}\ifx\@first\@empty\else\big\{#1\big\}\fi\ifx\@second\@empty\else\sqb{#2}\fi\ifx\@third\@empty\else({#3})\fi}%

\newcommand{\alpharesponsibilityCBot}[3]{\def\@first{#1}\def\@second{#2}\def\@third{#3}\alpha_{R_{c}}^{\bot}\ifx\@first\@empty\else\big\{#1\big\}\fi\ifx\@second\@empty\else\sqb{#2}\fi\ifx\@third\@empty\else({#3})\fi}%

\newcommand{\alpharesponsibilitySC}[3]{\def\@first{#1}\def\@second{#2}\def\@third{#3}\alpha_{R_{s\mskip-1mu c}}\ifx\@first\@empty\else\big\{#1\big\}\fi\ifx\@second\@empty\else\sqb{#2}\fi\ifx\@third\@empty\else({#3})\fi}%

\newcommand{\alpharesponsibilitySCF}[3]{\def\@first{#1}\def\@second{#2}\def\@third{#3}\alpha_{R_{s\mskip-1mu c}}^{F}\ifx\@first\@empty\else\big\{#1\big\}\fi\ifx\@second\@empty\else\sqb{#2}\fi\ifx\@third\@empty\else({#3})\fi}%

\newcommand{\alpharesponsibilitySCH}[3]{\def\@first{#1}\def\@second{#2}\def\@third{#3}\alpha_{R_{s\mskip-1mu c}}^{H}\ifx\@first\@empty\else\big\{#1\big\}\fi\ifx\@second\@empty\else\sqb{#2}\fi\ifx\@third\@empty\else({#3})\fi}%

\newcommand{\alpharesponsibilitySCTop}[3]{\def\@first{#1}\def\@second{#2}\def\@third{#3}\alpha_{R_{s\mskip-1mu c}}^{\top}\ifx\@first\@empty\else\big\{#1\big\}\fi\ifx\@second\@empty\else\sqb{#2}\fi\ifx\@third\@empty\else({#3})\fi}%

\newcommand{\alpharesponsibilitySCBot}[3]{\def\@first{#1}\def\@second{#2}\def\@third{#3}\alpha_{R_{s\mskip-1mu c}}^{\bot}\ifx\@first\@empty\else\big\{#1\big\}\fi\ifx\@second\@empty\else\sqb{#2}\fi\ifx\@third\@empty\else({#3})\fi}%

\newcommand{\alpharesponsibilityPearl}[3]{\def\@first{#1}\def\@second{#2}\def\@third{#3}\alpha_{R_{c}}^{FF}\ifx\@first\@empty\else\big\{#1\big\}\fi\ifx\@second\@empty\else\sqb{#2}\fi\ifx\@third\@empty\else({#3})\fi}%

\makeatother

\makeatletter
\newcommand{\alphaHistory}[1]{\def\@first{#1}\alpha^{H}\ifx\@first\@empty\else({#1})\fi}%
\newcommand{\gammaHistory}[1]{\def\@first{#1}\gamma^{H}\ifx\@first\@empty\else({#1})\fi}%
\newcommand{\alphaFuture}[1]{\def\@first{#1}\alpha^{F}\ifx\@first\@empty\else({#1})\fi}%
\newcommand{\gammaFuture}[1]{\def\@first{#1}\gamma^{F}\ifx\@first\@empty\else({#1})\fi}%
\makeatother

\makeatletter
\newcommand{\alphaC}[2]{\def\@first{#1}\def\@second{#2}\alpha_{c}\ifx\@first\@empty\else\{{#1}\}\fi\ifx\@second\@empty\else\,#2\fi}%
\newcommand{\alphaSC}[2]{\def\@first{#1}\def\@second{#2}\alpha_{s\mskip-1mu c}\ifx\@first\@empty\else\{{#1}\}\fi\ifx\@second\@empty\else\,#2\fi}%

\newcommand{\alphaCbar}[2]{\def\@first{#1}\def\@second{#2}\bar{\alpha}_{c}\ifx\@first\@empty\else\{{#1}\}\fi\ifx\@second\@empty\else\,#2\fi}%
\newcommand{\alphaSCbar}[2]{\def\@first{#1}\def\@second{#2}\bar{\alpha}_{s\mskip-1mu c}\ifx\@first\@empty\else\{{#1}\}\fi\ifx\@second\@empty\else\,#2\fi}%
\makeatother

\newcommand{\abstractdomain}{\mathcal{D}}
\newcommand{\userspec}[1]{#1^{u}}
\makeatletter
\newcommand{\nonomniscientcognizance}[2]{\def\@first{#1}\def\@second{#2}\mathbb{C}_{n\mskip-1muo}\ifx\@first\@empty\else\sqb{#1}\fi\ifx\@second\@empty\else({#2})\fi}%
\newcommand{\abstractinvariancesemantics}[1]{\def\@first{#1}\overline{\mathcal{S}}\ifx\@first\@empty\else\sqb{#1}\fi}%
\newcommand{\abstractobservation}[3]{\def\@first{#1}\def\@second{#2}\def\@third{#3}\overline{\mathsf{O}}\ifx\@first\@empty\else\{#1\}\fi\ifx\@second\@empty\else\sqb{#2}\fi\ifx\@third\@empty\else({#3})\fi}%
\makeatother
\newcommand{\abstracttracepropertyT}{\overline{\tracepropertyT}}
\newcommand{\abstracttraceproperty}{\overline{\traceproperty}}
\newcommand{\abstracttrace}{\overline{\trace}}
\newcommand{\abstractbehavior}{\overline{\behavior}}
\newcommand{\sqsubsetneq}{\mathrel{\rlap{\raisebox{-0.5ex}[0pt][0pt]{\tiny$\mskip7mu\scriptscriptstyle\shortmid$}}{\sqsubseteq}}}
\newcommand{\tracesemantics}{\mathcal{S}}
\makeatletter
\newcommand{\action}[2]{\overline{\mathsf{A}}\def\@actionparams{#1#2}\ifx\@actionparams\@empty\else(#1,#2)}
\makeatother
\makeatletter
\def\dashsfill{$\m@th\mathord-\mkern-7mu
     \cleaders\hbox{$\!\mathord-\!$}\hfill
     \mkern-7mu\mathord\rightarrow$}
\makeatother
\makeatletter
\newcommand{\transition}[1]{\sbox{\@tempboxa}{\hskip0.5em\raisebox{-0.25ex}{\small$#1$}\hskip0.66em}\ensuremath{\stackrel{\box\@tempboxa}{\hbox to \wd\@tempboxa{\dashsfill}}}}
\makeatother
\newcommand{\abstractinvariant}{\overline{\mathsf{I}}}

\begin{document}

%
\title{Responsibility Analysis\\ by Abstract Interpretation}
\titlerunning{Responsibility Analysis by Abstract Interpretation}
%
\author{Chaoqiang Deng \and
Patrick Cousot}
\authorrunning{C. Deng \and P. Cousot}
\institute{Computer Science Department, New York University, USA\\
\email{\{deng,pcousot\}@cs.nyu.edu}}
\maketitle              
\begin{abstract}
Given a behavior of interest in the program, statically determining the corresponding responsible entity is a task of critical importance, especially in program security. Classical static analysis techniques (e.g. dependency analysis, taint analysis, slicing, etc.) assist programmers in narrowing down the scope of responsibility, but none of them can explicitly identify the responsible entity. Meanwhile, the causality analysis is generally not pertinent for analyzing programs, and the structural equations model (SEM) of actual causality misses some information inherent in programs,
making its analysis on programs imprecise.
In this paper, a novel definition of responsibility based on the abstraction of event trace semantics is proposed, which can be applied in program security and other scientific fields. Briefly speaking, an entity $\entityresp$ is responsible for behavior $\behavior$, if and only if $\entityresp$ is free to choose its input value, and such a choice is the first one that ensures the occurrence of $\behavior$ in the forthcoming execution. Compared to current analysis methods, the responsibility analysis is more precise. In addition, our definition of responsibility takes into account the cognizance of the observer, which, to the best of our knowledge, is a new innovative idea in program analysis. 


\end{abstract}

\section{Introduction}

For any behavior of interest, especially potentially insecure behaviors in the program, it is essential to determine the corresponding responsible entity, or say, the root cause. Contrary to accountability mechanisms \cite{DBLP:journals/cacm/WeitznerABFHS08,DBLP:conf/esorics/JagadeesanJPR09,DBLP:conf/uss/FranklePSGW18} that track down perpetrators after the fact, the goal of this paper is to detect the responsible entity and configure its permission before deploying the program, which is important for safety and security critical systems. Due to the massive scale of modern software, it is virtually impossible to identify the responsible entity manually. 
The only solution is to design a static analysis of responsibility, which can examine all possible executions of a program without executing them. 

The cornerstone of designing such an analysis is to define responsibility in programming languages. It is surprising to notice that, although the concepts of causality and responsibility have been long studied in various contexts (law sciences \cite{Sliedregt-Criminal-Responsibility-2012}, artificial intelligence \cite{Pearl-Causality-2012}, statistical and quantum mechanics, biology, social sciences, etc. \cite{BeebeeHM-Causality-2009}), none of these definitions is fully pertinent for programming languages. Take the actual cause \cite{Halpern-Pearl-2001,Halpern-Pearl-2005} as an example, its structural equations model (SEM) \cite{Westland-SEM-2015} is not suitable for representing programs: the value of each endogenous variable in the model is fixed once it is set by the equations or some external action, while the value of program variables can be assigned for unbounded number of times during the execution. In addition, the SEM cannot make use of the temporal information or whether an entity is free to make choices, which plays an indispensable role in determining responsibility. 

There do exist techniques analyzing the influence relationships in programs, such as dependency analysis \cite{DBLP:conf/popl/AbadiBHR99,DBLP:journals/mscs/CheneyAA11,DBLP:conf/esop/UrbanM18}, taint analysis \cite{DBLP:conf/ecoop/PistoiaFKS05} and program slicing \cite{DBLP:journals/tse/Weiser84}, which help in narrowing down the scope of possible locations of responsible entity. However, no matter whether adopting semantic or syntactic methods, these techniques are not precise enough to explicitly identify responsibility. 

To solve the above problems, we propose a novel definition of responsibility based on the event trace semantics, which is expressive and generic to handle computer programs and other scientific fields. Roughly speaking, an entity $\entityresp$ is responsible for a given behavior $\behavior$ in a certain trace, if and only if $\entityresp$ can choose various values at its discretion (e.g. inputs from external subjects), and such a choice is the first one that guarantees the occurrence of $\behavior$ in that trace. Such a definition of responsibility is an abstract interpretation \cite{CousotC77,CousotC79} of event trace semantics, taking into account both the temporal ordering of events and the information regarding whether an entity is free to choose its value. Moreover, an innovative idea of cognizance is adopted in this definition, which allows analyzing responsibility from the perspective of various observers.
Compared to current techniques, our definition of responsibility is more generic and precise.



The applications of responsibility analysis are pervasive. Although an implementation of responsibility analyzer is not provided here, we have demonstrated its effectiveness by examples including access control, ``negative balance" and information leakage. In addition, due to the page limit, a sound framework of abstract responsibility analysis is sketched
\iflong
in the appendix \ref{sec:abstract-responsibility},
\else
in the extended version of this paper \cite{SAS-2019-extended-version},
\fi
which is the basis of implementing a responsibility analyzer. It is guaranteed that the entities that are found definitely responsible in the abstract analysis are definitely responsible in the concrete, while those not found potentially responsible in the abstract analysis are definitely not responsible in the concrete.





To summarize, the main contributions of this work are: (1) a completely new definition of responsibility, which is based on the abstract interpretation of event trace semantics, (2) the adoption of observers' cognizance in program analysis for the first time, (3) various examples of responsibility analysis, and (4) a sound framework for the abstract static analysis of responsibility.
%

In the following, section \ref{sec:informal-definition} discusses the distinctions between responsibility and current techniques via an example, and sketches the framework of responsibility analysis. Section \ref{sec:formal-definition} formally defines responsibility as an abstraction of event trace semantics. Section \ref{sec:application} exemplifies the applications of responsibility analysis. 
Section \ref{sec:related-work} summarizes the related work. 
\section{A Glance at Responsibility\label{sec:informal-definition}}


Given a behavior of interest (e.g. security policy violation), the objective of responsibility analysis is to automatically determine which entity in the system has the primary control over that behavior. Then security admins could decide either to keep or to deny the responsible entity's permission to perform the behavior of interest. Take the information leakage in a social network as an example: if the information's owner is responsible for the leakage (e.g. a user shares his picture with friends), then it is safe to keep its permission to perform such a behavior; otherwise, if anyone else is responsible for the leakage, it could be a malicious attacker and its permission to do so shall be removed. Such human decisions can only be done manually and are beyond the scope of this paper. In addition, it is worthwhile to note that responsibility analysis is not the same as program debugging, since the analyzed program code is presumed to be unmodifiable and the only possible change is on the permissions granted to entities in the system. 

In order to give an informal introduction to responsibility, as well as its main distinctions with dependency, causality and other techniques in detecting causes, this section starts with a simple example, which is used throughout the paper.

\subsection{Discussion of an Access Control Program Example\label{subsec:example}}



\begin{example}[Access Control]\label{ex:access-control}
Consider the program in Fig. \ref{fig:access-control}, which essentially can be interpreted as an access control program for an object $o$ (e.g. a secret file), such that $o$ can be read if and only if both two admins approve the access and the permission type of $o$ from system settings is greater than or equal to ``read'': the first two inputs correspond to the decisions of two independent admins, where 1 represents approving the access to $o$, and 0 represents rejecting the access; the third input stored in $\footnotesize\texttt{typ}$ represents the permission type of $o$ specified in the system settings, where 1 represents ``read'', 2 represents ``read and write''
(this is similar to the file permissions system in Linux, but is simplified for the sake of clarity); 
by checking the value of $\footnotesize\texttt{acs}$ at line 10, the assertion can guarantee both admins approve the access and the permission type of $o$ is at least 1.
\qed
\end{example}

\begin{figure}[h]
\footnotesize
\ttfamily
\vskip-1.5em
\begin{lstlisting}
1:  apv = 1; //1: Approval, 0: Rejection
2:  i1 = input_1(); //Input 0 or 1 from 1st admin
3:  if (i1 == 0) {
4:    apv = 0; }
5:  i2 = input_2(); //Input 0 or 1 from 2nd admin
6:  if (apv != 0 && i2 == 0) {
7:    apv = 0; }
8:  typ = input_3(); //Input 1 or 2 from system settings
9:  acs = apv * typ;
10: assert(acs >= 1); //Check if the read access is granted
11: /* Read an object o here */
\end{lstlisting}
\vskip-.5em
\caption{Access Control Program Example\label{fig:access-control}}
\vskip-2em
\end{figure}

Here the question we are interested is: when the assertion fails (referred as ``Read Failure'' in the following, i.e. the read access to $o$ fails to be granted), which entity (entities) in the program shall be responsible? 
The literature has several possible answers. By the definition of dependency (\cite{DBLP:conf/popl/AbadiBHR99,DBLP:journals/mscs/CheneyAA11,DBLP:conf/esop/UrbanM18}), the value of $\footnotesize\texttt{acs}$ depends on the value of $\footnotesize\texttt{apv}$ and $\footnotesize\texttt{typ}$, which further depend on all three inputs. That is to say, the read failure depends on all variables in the program, thus program slicing techniques (both syntactic slicing \cite{DBLP:journals/tse/Weiser84} and semantic slicing \cite{DBLP:conf/sas/Rival05}) would take the whole program as the slice related with read failure. Such a slice is useful in debugging in the sense that it rules out parts of the program that are completely irrelevant with the failure, and modifying any code left in the slice may prevent the failure, e.g. replacing ``$\footnotesize\texttt{acs=apv*typ}$'' with ``$\footnotesize\texttt{acs=2}$'' trivially fixes the read failure problem. However, this paper presumes the program code to be unmodifiable, hence a statement like ``$\footnotesize\texttt{acs=apv*typ}$'', which is fully controlled by others and acts merely as the intermediary between causes and effects, 
shall not be treated as responsible. In addition, the third input (i.e. the system setting of $o$'s permission type) is also included in the slice. Although it does affect $\footnotesize\texttt{acs}$'s value, it is not decisive in this case (i.e. no matter it is 1 or 2, it could not either enforce or prevent the failure). Therefore, the dependency analysis and slicing are not precise enough for determining responsibility.

Causation by counterfactual dependency \cite{Lewis-causation-1973} examines the cause in every single execution and excludes non-decisive factors (e.g. the third input in this example), but it is too strong in some circumstances. For example, in an execution where both the first two inputs are 0, neither of them would be determined as the cause of read failure, because if one input is changed to value 1, the failure would still occur due to the other input 0.

Actual cause introduced in \cite{Halpern-Pearl-2001,Halpern-Pearl-2005}  is based on the structural equations model (SEM) \cite{Westland-SEM-2015}, and extends the basic notion of counterfactual dependency to allow ``contingent dependency''. For this example here, it is straightforward to create a SEM to represent the access control program (although it is not always the case): three inputs are represented by exogenous variables, and five program variables are represented by endogenous variables, in which the value of $\footnotesize\texttt{apv}$ is $\footnotesize\texttt{i1*i2}$. Consider an execution where both the first two inputs are 0, no matter what value the third input takes, the actual causes of read failure (i.e. $\footnotesize\texttt{acs<1}$) would be determined as ``$\footnotesize\texttt{i1=0}$'', ``$\footnotesize\texttt{i2=0}$'', ``$\footnotesize\texttt{apv=0}$'' and ``$\footnotesize\texttt{acs=0}$'', since the failure counterfactually depends on each of them under certain contingencies. Thus, both two admins are equally determined as causes of failure, as well as two intermediary variables. This structural-model method has allowed for a great progress in causality analysis, and solved many problems of previous approaches. However, as an abstraction of concrete semantics, the SEM unnecessarily misses too much information, including the following three important points.

\textbf{(P1)\label{time}} \textit{Time (i.e. the temporal ordering of events) should be taken into account.} For example, the SEM does not keep the temporal ordering of first two inputs (i.e. the information that ``$\footnotesize\texttt{i1=0}$'' occurs before ``$\footnotesize\texttt{i2=0}$'' is missed), hence it determines both of them equally as the cause of assigning $\footnotesize\texttt{0}$ to $\footnotesize\texttt{apv}$, further as the cause of read failure. However, in the actual execution where first two inputs are 0, the first input already decides the value of $\footnotesize\texttt{apv}$ before the second input is entered and the assignment at line 7 is not even executed, thus it is unnecessary to take the second input as a cause of assigning $\footnotesize\texttt{0}$ to $\footnotesize\texttt{apv}$ or the read failure. To deal with this difficulty, Pearl's solution is to modify the model and introduce new variables \cite{DBLP:conf/ijcai/ChenPB16} to distinguish whether $\footnotesize\texttt{apv}$ is assigned by $\footnotesize\texttt{i1}$ or $\footnotesize\texttt{i2}$. However, a much simpler method is to keep the temporal ordering of events, such that only the first event that ensures the behavior of interest is counted as the cause. Therefore, in an execution where both the first two inputs are 0, the first input ensures the read failure before the second input is entered, hence only the first input is responsible for failure; meanwhile, in another execution where the first input is 1 and the second one is 0, the second input is the first and only one that ensures the failure hence shall take the responsibility.

\textbf{(P2)\label{free-choices}} \textit{The cause must be free to make choices.} For example, $\footnotesize\texttt{acs=0}$ is determined as an actual cause of read failure, based on the reasoning that if the endogenous variable $\footnotesize\texttt{acs}$ in SEM is assigned a different value, say $\footnotesize\texttt{2}$, then the read failure would not have occurred. But such a reasoning ignores a simple fact that $\footnotesize\texttt{acs}$ is not free to choose its value and acts merely as an intermediary between causes and effects. Thus, only entities that are free to make choices can possibly be causes, and they include but are not limited to user inputs, system settings, files read, parameters of procedures or modules, returned values of external functions, variable initialization, random number generations and the parallelism. To be more accurate, it is the external subject (who does the input, configures the system settings, etc.) that is free to make choices, but we say that entities like user inputs are free to make choices, as an abuse of language.

\textbf{(P3)\label{cognizance}} \textit{It is necessary to specify ``to whose cognizance~/ knowledge'' when identifying the cause.} All the above reasoning on causality is implicitly based on an omniscient observer's cognizance (i.e. everything that occurred is known), yet it is non-trivial to consider the causality to the cognizance of a non-omniscient observer. Reconsider the access control program example, and suppose we adopt the cognizance of the second admin who is in charge of the second input. If she/he is aware that the first input is already 0, she/he would not be responsible for the failure; otherwise she/he does not know whether the first input is 0 or 1, then she/he is responsible for ensuring the occurrence of failure. In most cases, the cognizance of an omniscient observer will be adopted, but not always.

\subsection{An Informal Definition of Responsibility\label{subsec:informal-definition}}

To take the above three points into account and build a more expressive framework, this paper proposes \textit{responsibility}, whose informal definition is as follows.


\begin{mdframed}[innertopmargin=4pt,innerbottommargin=4pt] 
\begin{definition}[\textbf{Responsibility, informally}\label{def:informal-resp}]
To the cognizance of an observer, the entity $\entityresp$ is responsible for a behavior $\behavior$ of interest in a certain execution, if and only if, according to the observer's observation, $\entityresp$ is free to choose its value, and such a choice is the first one that guarantees the occurrence of $\behavior$ in that execution.
\end{definition}
\end{mdframed}


It is worth mentioning that, for the whole system whose semantics is a set of executions, there may exist more than one entities that are responsible for $\behavior$. Nevertheless, in every single execution where $\behavior$ occurs, there is only one entity that is responsible for $\behavior$. To decide which entity in an execution is responsible, the execution alone is not sufficient, and it is necessary to reason on the whole semantics to exhibit the entity's ``free choice'' and guarantee of $\behavior$. Thus, responsibility is not a trace property (neither safety nor liveness property). 

To put such a definition into effect, our framework of responsibility analysis is designed as Fig. \ref{fig:analysis-framework}, which essentially consists of three components: (1) \textit{System semantics}, i.e. the set of all possible executions, each of which can be analyzed individually. (2) \textit{A lattice of system behaviors of interest}, which is ordered such that the stronger a behavior is, the lower is its position in the lattice. (3) \textit{An observation function for each observer}, which maps every (probably unfinished) execution to a behavior in the lattice that is guaranteed to occur, even though such a behavior may have not occurred yet. These three components are formally defined in section \ref{sec:formal-definition}, and their abstractions are sketched 
\iflong
in the appendix \ref{sec:abstract-responsibility}.
\else
in \cite{SAS-2019-extended-version}.
\fi

\vskip-1.5em
\begin{figure}[h]
\centering
  \includegraphics[width=.85\linewidth]{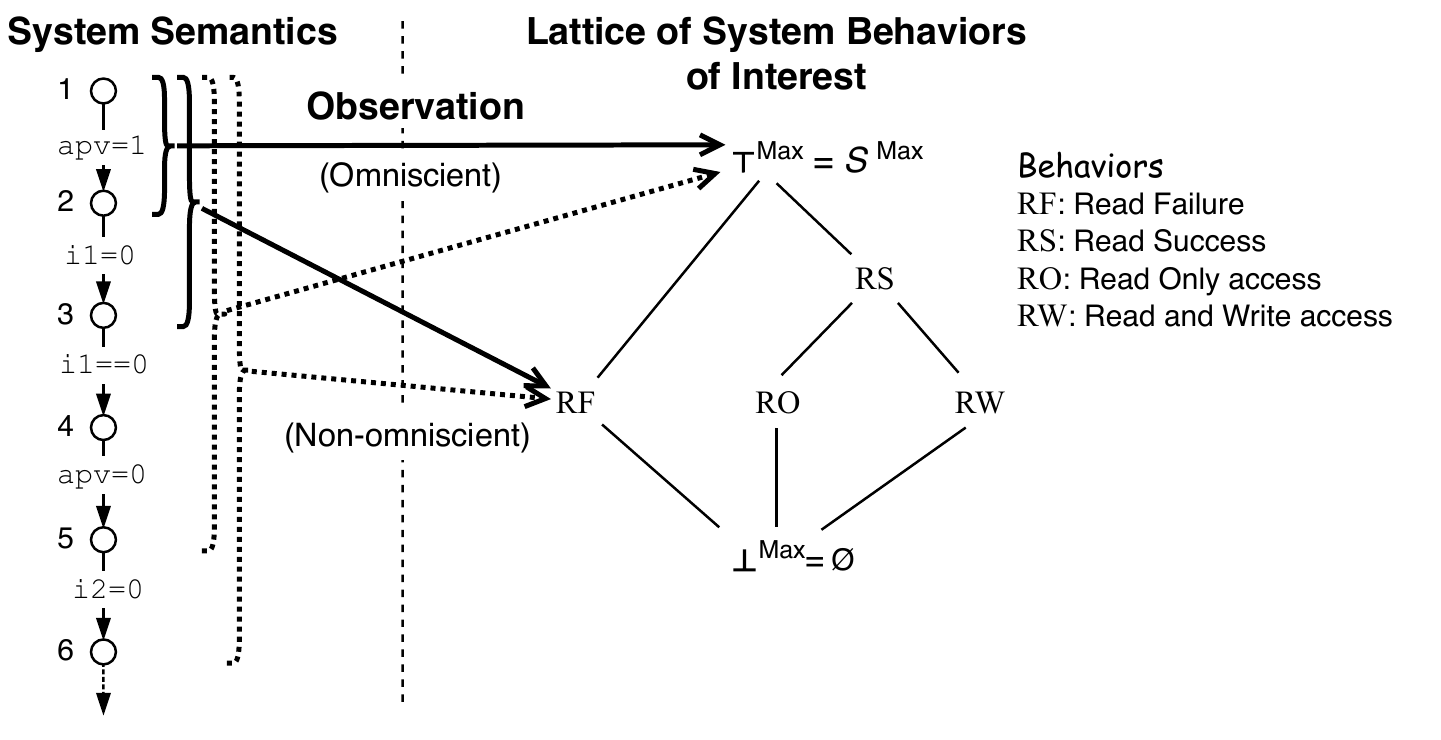}
  \vskip-1em
\caption{Framework of Responsibility Analysis for Example \ref{ex:access-control}}
\label{fig:analysis-framework}
\end{figure}
\vskip-1.5em

In this framework, if an observer's observation finds that the guaranteed behavior grows stronger after extending an execution, then the extension part of execution must be responsible for ensuring the occurrence of the stronger behavior. Consider the example in Fig. \ref{fig:analysis-framework} which sketches the analysis for a certain execution of the access control program. Suppose $\maximaltracepropertytop$ in the lattice represents ``not sure if the read access fails or not'' and \texttt{RF} represents the behavior of read failure, whose formal definitions are given in section \ref{subsec:lattice-trace-properties}. The solid arrow from executions to the lattice stands for the observation of an omniscient observer, while the dashed arrow stands for the observation of the second admin who is unaware of the first input. As illustrated in the figure, the omniscient observer finds that the execution from point 1 to point 2 can guarantee only $\maximaltracepropertytop$, while the stronger behavior \texttt{RF} is guaranteed if the execution reaches point 3. Thus, to the cognizance of the omniscient observer, ``$\footnotesize\texttt{i1=0}$'' between point 2 and 3 is responsible for the read failure. Meanwhile, the second admin observes that all the executions upto point 5 guarantee $\maximaltracepropertytop$, and \texttt{RF} is guaranteed only after point 6 is reached. Hence, to the cognizance of the second admin, ``$\footnotesize\texttt{i2=0}$'' between point 5 and point 6 is responsible for the read failure.
For the sake of completeness, the entire desired analysis result for Example \ref{ex:access-control} is included in the following.




\begin{example}[Access Control, Continued]\label{ex:access-control-responsibility}To the cognizance of an omniscient observer:
for any execution, if the first input is 0, no matter what the other two inputs are, only the first admin is responsible for the read failure; if the first input is 1 and the second one is 0, the second admin is responsible.


To the cognizance of the second admin, two cases need to be considered separately. If she/he is aware of the input of first admin, the analysis result is exactly the same as the omniscient observer. Otherwise, she/he does not know the first input: in every execution where the second input is 0, the second admin is responsible, no matter what the first and third input are; in every execution where the second input is 1, nobody shall be responsible for the failure, since whether the failure occurs or not is uncertain from the second admin's perspective.
\qed
\end{example}

After finishing responsibility analysis, it is time for the security admin to configure permissions granted to each responsible entity at her/his discretion. If the behavior of interest is desired or the responsible entity is authorized, the permissions granted to the responsible entity can be kept. On the contrary, if that behavior is undesired or it is against the policy for the responsible entity to control it, the permissions granted to the responsible entity shall be confined. For instance, in the access control program, if the first two inputs are from admins who are authorized to control the access, their permissions to input 0 and 1 can be kept; if those two inputs come from ordinary users who have no authorization to deny other users' access, their permissions to input 0 shall be removed.
\section{Formal Definition of Responsibility\label{sec:formal-definition}}

\vskip-.3em
In order to formalize the framework of responsibility analysis, this section introduces event traces to represent the system semantics, builds a lattice of system behaviors by trace properties, proposes an observation function that derives from the observer's cognizance and an inquiry function on system behaviors. Furthermore, this section formally defines responsibility as an abstraction of system semantics, using the observation function. To strengthen the intuition of responsibility analysis, the analysis of Example \ref{ex:access-control} will be illustrated step by step.



\subsection{System Semantics\label{subsec:event-model}}

Generally speaking, no matter what system we are concerned with and no matter which programming language is used to implement that system, the system's semantics can be represented by event traces. 
\vskip-1em

\vskip-.5em
\subsubsection{Event Trace}
In general, an \textit{event} could be used to represent any action in the system, such as ``input an integer'', ``assign a value to a variable'', or even ``launch the program''. Take the classic While programming language as an example, there are only three types of events: skip, assignment, and Boolean test. In order to make the definition of responsibility as generic as possible, here we do not adopt a specific programming language or restrict the range of possible events.

A \textit{trace} $\trace$ is a sequence of events that represents an execution of the system, 
and its length $|\trace|$ is the number of events in $\trace$. If $\trace$ is infinite, then its length $|\trace|$ is denoted as $\infty$. A special trace is the empty trace $\emptytrace$ , whose length is 0. A trace $\trace$ is $\isprefix$ - less than or equal to another trace $\trace'$, if and only if, $\trace$ is a prefix of $\trace'$. The \textit{concatenation} of a finite trace $\trace$ and an event $\event$ is simply defined by juxtaposition $\trace\event$, and the \textit{concatenation} of a finite traces $\trace$ and another (finite or infinite) trace $\trace'$ is denoted as $\trace\trace'$. 

\begin{eqntabular}[fl]{@{}r@{~}c@{~}l@{~}c@{~}l}
\event &\in&\setofallevents&&\renumber{event}\\
\trace &\in&\setofallnonemptytrace&\triangleq& \bigcup_{n \geqslant 1}\{[0, n-1] \mapsto \setofallevents\}\cup\{\naturalnum\mapsto\setofallevents\} \renumber{nonempty trace}\\
\trace &\in&\setofalltraces&\triangleq& \;\{\emptytrace\}\cup\setofallnonemptytrace \renumber{empty or nonempty trace}\\
\trace &\isprefix&\trace'&\triangleq& \;|\trace| \leqslant |\trace'| \wedge \forall\, 0 \leqslant i \leqslant |\trace| - 1: \trace_{i} = \trace'_{i}\renumber{prefix ordering of traces}
\end{eqntabular}
The function $\prefixesofset{P}$ returns the prefixes of every trace in the set $P$ of traces.

\begin{eqntabular}[fl]{@{}r@{~}c@{~}l}
\prefixesofset{} &\in& \wp(\setofalltraces)\mapsto\wp(\setofalltraces) \renumber{prefixes of traces}\\
\prefixesofset{P}&\triangleq& \{\trace'\in\setofalltraces\mid\exists\trace\in P.\;\trace'\isprefix\trace\} \nonumber
\end{eqntabular}
\vskip-2em

\vskip-2em
\subsubsection{Trace Semantics}\label{sec:System-Semantics}

For any system that we are concerned with, its \textit{maximal trace semantics}, denoted as $\maximalsemantics{}\in\wp(\setofalltraces)$, is the set of all possible maximal traces of that system. Especially, the maximal trace semantics of an empty program is $\{\emptytrace\}$. Correspondingly, the \textit{prefix trace semantics} $\prefixsemantics{}\in\wp(\setofalltraces)$ is the set of all possible prefix traces, which is an abstraction of maximal trace semantics via $\prefixesofset{}$, i.e. $\prefixsemantics{}=\prefixesofset{\maximalsemantics{}}$. Besides, a trace $\trace$ is said to be \textit{valid} in the system, if and only if $\trace\in\prefixsemantics{}$. Obviously, both maximal and prefix trace semantics do preserve the temporal ordering of events, which is missed by the SEM.


\begin{figure}[h]
\centering
  \includegraphics[width=.85\linewidth]{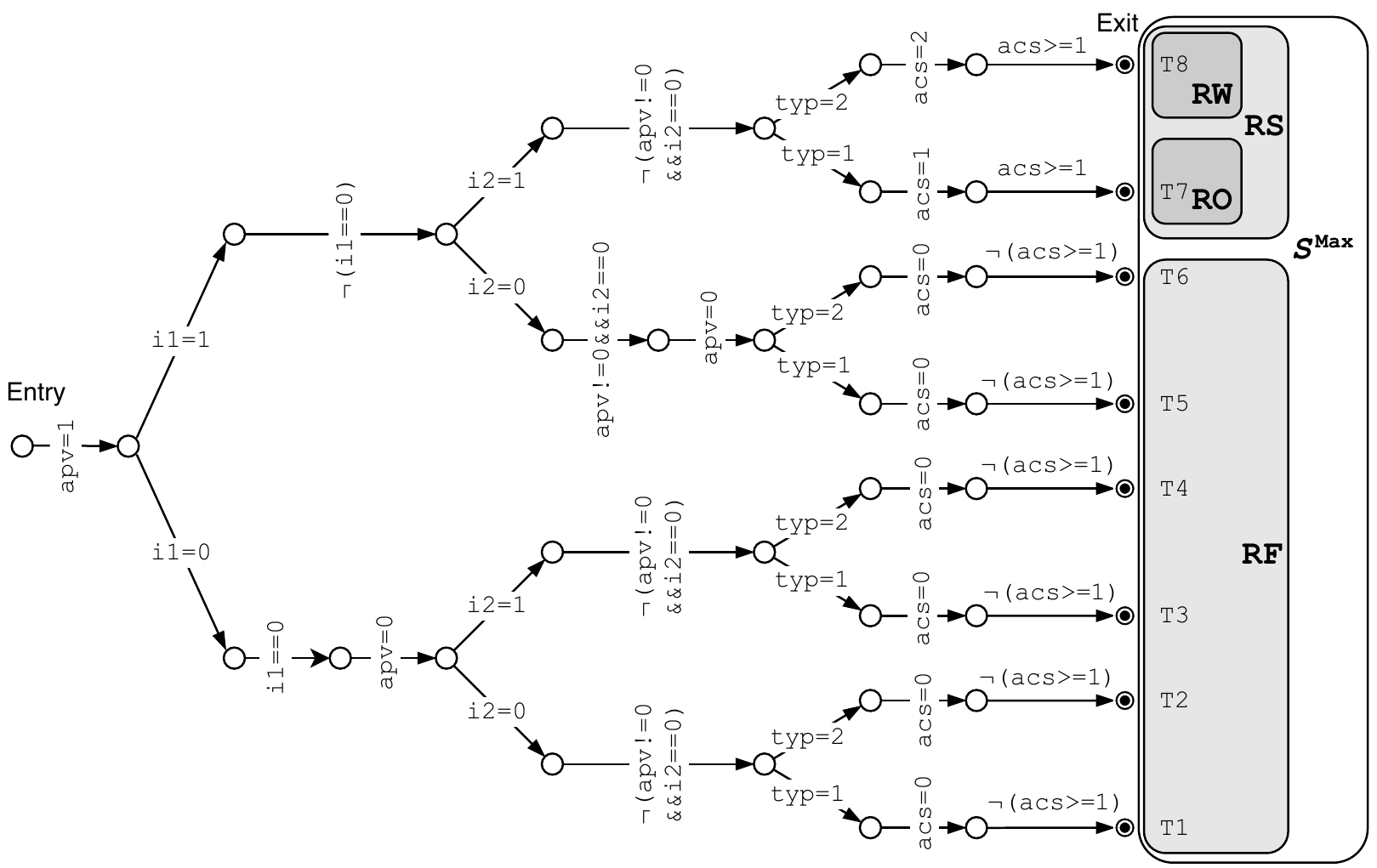}
\caption{Trace Semantics and Properties of Example \ref{ex:access-control}}
\label{fig:semantics-access-control}
\end{figure}

\begin{example}[Access Control, Continued]\label{ex:access-control-semantics}
For the program in Fig. \ref{fig:access-control}, only two types of events are used: assignment (e.g. $\footnotesize\texttt{apv=1}$) and Boolean test (e.g. $\footnotesize\texttt{i1==0}$ and ${\scriptstyle\neg}\footnotesize\texttt{(acs>=1)}$, where $\textstyle\neg$ denotes the failure of a Boolean test). To clarify the boundary among events, the triangle $\gap$ is used in the following to separate events in the trace. The access control program has three inputs, each of which has two possible values, thus its maximal trace semantics $\maximalsemantics{}$ consists of 8 traces (${\footnotesize\texttt{T1}}$-${\footnotesize\texttt{T8})}$, each of which is represented as a path in Fig. \ref{fig:semantics-access-control} starting at the entry point of program and finishing at the exit point. E.g. {\footnotesize $\texttt{T1}=$ $\texttt{apv=1}$ $\gap$ $\texttt{i1=0}$ $\gap$ $\texttt{i1==0}$ $\gap$ $\texttt{apv=0}$ $\gap$ $\texttt{i2=0}$ $\gap$ ${\scriptstyle\neg}\texttt{(apv!=0\&\&i2==0)}$ $\gap$ $\texttt{typ=1}$ $\gap$ $\texttt{acs=0}$ $\gap$ ${\scriptstyle\neg}\texttt{(acs>=1)}$} denotes the maximal execution where the first two inputs are 0 and the third input is 1. Meanwhile, the prefix trace semantics $\prefixsemantics{}=\prefixesofset{\maximalsemantics{}}$ are represented by the paths that start at the entry point and stop at any point (including the entry point for the empty trace $\emptytrace$).
\qed\end{example}

\subsection{Lattice of System Behaviors of Interest\label{subsec:lattice-trace-properties}}

\subsubsection{Trace Property} A \textit{trace property} is a set of traces in which a given property holds. Most behaviors of a given system, if not all, can be represented as a maximal trace property $\traceproperty\in\wp(\maximalsemantics{})$. 
\vskip-.5em

\begin{example}[Access Control, Continued]\label{ex:access-control-properties}
As illustrated in Fig. \ref{fig:semantics-access-control}, the behavior ``Read Failure'' \texttt{RF} 
is represented as a set of maximal traces such that the last event is ${\scriptstyle\neg}{\footnotesize \texttt{(acs>=1)}}$, i.e. \texttt{RF} = \{$\trace\in\maximalsemantics{}\mid$ $\trace_{|\trace|-1}$ = ${\scriptstyle\neg}{\footnotesize \texttt{(acs>=1)}}$\} = \{{\footnotesize\texttt{T1}}, {\footnotesize\texttt{T2}}, {\footnotesize\texttt{T3}}, {\footnotesize\texttt{T4}}, {\footnotesize\texttt{T5}}, {\footnotesize\texttt{T6}}\}; the behavior ``Read Success'' \texttt{RS}  (i.e. the read access succeeds to be granted) is the complement of \texttt{RF}, i.e. \texttt{RS} = $\maximalsemantics{}\backslash\texttt{RF}$ = \{{\footnotesize\texttt{T7}}, {\footnotesize\texttt{T8}}\}, whose subset \texttt{RO} = \{{\footnotesize\texttt{T7}}\} and \texttt{RW} = \{{\footnotesize\texttt{T8}}\} represent stronger properties ``Read Only access is granted'' and ``Read and Write access is granted'', respectively.
\qed\end{example}
\vskip-1.5em

\vskip-1em
\subsubsection{Complete Lattice of Maximal Trace Properties of Interest\label{sec:Lattice-of-Maximal-Trace-Properties}} 
We build a complete lattice of maximal trace properties, each of which represents a behavior of interest.
Typically, such a lattice is of form $\sextuple{\latticeofmaximalproperties}{\maximaltracepropertylessthan}{\maximaltracepropertytop}{\maximaltracepropertybot}{\maximaltracepropertyjoin}{\maximaltracepropertymeet}$, where 

\noindent
\begin{itemize}[leftmargin=*]
\setlength\itemsep{.2em}
\item[--] $\latticeofmaximalproperties\in\wp(\wp(\setofalltraces))$ is a set of behaviors of interest, each of which is represented by a maximal trace property;

\item[--] $\maximaltracepropertytop$  $=$ $\maximalsemantics{}$, i.e. the top  is the weakest maximal trace property which holds in every valid maximal trace; 

\item[--] $\maximaltracepropertybot$ $=$ $\emptyset$, i.e. the bottom is the strongest property such that no valid trace has this property, hence it is used to represent the property of invalidity; 

\item[--] $\maximaltracepropertylessthan$ is the standard set inclusion operation;

\item[--] $\maximaltracepropertyjoin$ and $\maximaltracepropertymeet$ are join and meet operations, which might not be the standard $\cup$ and $\cap$, since $\latticeofmaximalproperties$ is a subset of $\wp(\maximalsemantics{})$ but not necessarily a sublattice.
\end{itemize}

For any given system, there is possibly more than one way to build the complete lattice of maximal trace properties, depending on which behaviors are of interest. A special case of lattice is the power set of maximal trace semantics, i.e. $\latticeofmaximalproperties$ $=$ $\wp(\maximalsemantics{})$, which can be used to examine the responsibility for every possible behavior in the system. However, in most cases, a single behavior is of interest, and it is sufficient to adopt a lattice with only four elements: $\behavior$ representing the behavior of interest, $\maximalsemantics{}\backslash\behavior$ representing the complement of the behavior of interest, as well as the top $\maximalsemantics{}$ and bottom $\emptyset$. Particularly, if $\behavior$ is equal to $\maximalsemantics{}$, i.e. every valid maximal trace in the system has this behavior of interest, then a trivial lattice with only the top and bottom is built, from which no responsibility can be found, making the corresponding analysis futile. 

\begin{example}[Access Control, Continued]\label{ex:access-control-lattice-maximal-trace-properties}
We assume that ``Read Failure'' is of interest, as well as the behavior of granting write access. As illustrated by the lattice in Fig.~\ref{fig:analysis-framework}, 
regarding whether the read access fails or not, the top $\maximaltracepropertytop$ is split into two properties ``Read Failure'' \texttt{RF} and ``Read Success'' \texttt{RS}, which are defined in Example \ref{ex:access-control-properties} such that \texttt{RF} $\maximaltracepropertyjoin$ \texttt{RS} = $\maximalsemantics{}$ and \texttt{RF} $\maximaltracepropertymeet$ \texttt{RS} = $\emptyset$.
Furthermore, regarding whether the write access is granted or not, \texttt{RS} is split into ``Read Only access is granted'' \texttt{RO} and ``Read and Write access is granted'' \texttt{RW}. Now every property of interest corresponds to an element in the lattice, and the bottom $\maximaltracepropertybot$ = $\emptyset$ is the meet $\maximaltracepropertymeet$ of \texttt{RF}, \texttt{RO} and \texttt{RW}. In addition, if ``Read Failure'' is the only behavior of interest, \texttt{RO}  and \texttt{RW} can be removed from the lattice. 
\qed\end{example}
\vskip-1.5em

%
%

\vskip-1em
\subsubsection{Prediction Abstraction} \label{subsubsec:prediction-abstraction} Although the maximal trace property is well-suited to represent system behaviors, it does not reveal the point along the maximal trace from which a property is guaranteed to hold later in the execution. Thus, we propose to abstract every maximal trace property $\traceproperty\in\latticeofmaximalproperties$ isomorphically into a set $\predictiontraceproperty$ of prefixes of maximal traces in $\traceproperty$, excluding those whose maximal prolongation may not satisfy the property $\traceproperty$. This abstraction is called \textit{prediction 
abstraction}, and $\predictiontraceproperty$ is a \textit{prediction trace property} corresponding to $\traceproperty$. It is easy to see that $\predictiontraceproperty$ is a superset of $\traceproperty$, and is not necessarily prefix-closed.

\begin{eqntabular}[fl]{@{}r@{~}c@{~}l}
\alphapropertytransform{\maximalsemantics{}}{} &\in& \wp(\setofalltraces)\mapsto\wp(\setofalltraces) \renumber{prediction abstraction}\\
\alphapropertytransform{\maximalsemantics{}}{\traceproperty}&\triangleq& \{\trace\in\prefixesofset{\traceproperty}\mid\forall\trace'\in\maximalsemantics{}.\; \trace\isprefix\trace'\Rightarrow\trace'\in\traceproperty\} \nonumber\\
\gammapropertytransform{\maximalsemantics{}}{} &\in& \wp(\setofalltraces)\mapsto\wp(\setofalltraces) \renumber{prediction concretization}\\
\gammapropertytransform{\maximalsemantics{}}{\mathcal{Q}}&\triangleq& \{\trace\in\mathcal{Q}\mid\trace\in\maximalsemantics{}\} = \mathcal{Q}\cap\maximalsemantics{} \nonumber
\end{eqntabular}

We have a Galois isomorphism between maximal trace properties and prediction trace properties:

\begin{eqntabular}{c}
\hskip-1.5em
\pair{\wp(\maximalsemantics{})}{\subseteq}\GaloiS{\alphapropertytransform{\maximalsemantics{}}{}}{\gammapropertytransform{\maximalsemantics{}}{}}\pair{\alphapropertysettransform{\maximalsemantics{}}{\wp(\maximalsemantics{})}}{\subseteq}\label{eq:alphapropertytransform}
\end{eqntabular}

\vskip-.8em
\noindent where the abstract domain is obtained by a function $\alphapropertysettransform{\maximalsemantics{}}{}\in \wp(\wp(\setofalltraces))\mapsto\wp(\wp(\setofalltraces))$, which is defined as $\alphapropertysettransform{\maximalsemantics{}}{\mathfrak{X}}\triangleq\{\alphapropertytransform{\maximalsemantics{}}{\traceproperty}\mid\traceproperty\in\mathfrak{X}\}$. The following lemma immediately follows from the definition of $\alphapropertytransform{\maximalsemantics{}}{}$.


\begin{lemma}\label{lemma:prediction-property} Given a prediction trace property $\predictiontraceproperty$ that corresponds to a maximal trace property $\traceproperty$, if a prefix trace $\trace$ belongs to $\predictiontraceproperty$, then $\trace$ guarantees the satisfaction of property $\traceproperty$ (i.e. every valid maximal trace that is greater than or equal to $\trace$ is guaranteed to have property $\traceproperty$).
\end{lemma}

\begin{example}[Access Control, Continued]\label{ex:access-control-property-abstraction}
By $\alphapropertytransform{}{}$, each behavior in the lattice $\latticeofmaximalproperties$ of Example \ref{ex:access-control-lattice-maximal-trace-properties} can be abstracted into a prediction trace property:

\begin{itemize}[leftmargin=*]
\setlength\itemsep{.2em}
\item[--] $\alphapropertytransform{\maximalsemantics{}}{\maximaltracepropertytop}$ $=\prefixsemantics{}$, i.e. every valid trace in $\prefixsemantics{}$ guarantees $\maximaltracepropertytop$.

\item[--] $\alphapropertytransform{\maximalsemantics{}}{\texttt{RF}}$ $=\{\trace\in\prefixsemantics{}\mid$ {\footnotesize $\texttt{apv=1}$ $\gap$ $\texttt{i1=0}$} $\isprefix\trace$ $\vee$ {\footnotesize $\texttt{apv=1}$ $\gap$ $\texttt{i1=1}$ $\gap$ ${\scriptstyle\neg}\texttt{(i1==0)}$ $\gap$ $\texttt{i2=0}$} $\isprefix\trace$ $\}$, i.e. for any valid trace, if at least one of first two inputs is 0, then it guarantees ``Read Failure'' $\texttt{RF}$.

\item[--] $\alphapropertytransform{\maximalsemantics{}}{\texttt{RS}}$ $=\{\trace\in\prefixsemantics{}\mid$ {\footnotesize $\texttt{apv=1}$ $\gap$ $\texttt{i1=1}$ $\gap$ ${\scriptstyle\neg}\texttt{(i1==0)}$ $\gap$ $\texttt{i2=1}$} $\isprefix\trace\}$,  i.e. for any valid trace, if first two inputs are 1, it guarantees ``Read Success'' $\texttt{RS}$.

\item[--] $\alphapropertytransform{\maximalsemantics{}}{\texttt{RO}}$ $=\{\trace\in\prefixsemantics{}\mid$ {\footnotesize $\texttt{apv=1}$ $\gap$ $\texttt{i1=1}$ $\gap$ ${\scriptstyle\neg}\texttt{(i1==0)}$ $\gap$ $\texttt{i2=1}$ $\gap$ ${\scriptstyle\neg}\texttt{(apv!=0\&\&i2==0)}$ $\gap$ $\texttt{typ=1}$} $\isprefix\trace\}$, i.e. for any valid trace, if first two inputs are 1 and the third input is 1, then it guarantees ``Read Only access is granted'' $\texttt{RO}$.

\item[--] $\alphapropertytransform{\maximalsemantics{}}{\texttt{RW}}$ $=\{\trace\in\prefixsemantics{}\mid$ {\footnotesize $\texttt{apv=1}$ $\gap$ $\texttt{i1=1}$ $\gap$ ${\scriptstyle\neg}\texttt{(i1==0)}$ $\gap$ $\texttt{i2=1}$ $\gap$ ${\scriptstyle\neg}\texttt{(apv!=0\&\&i2==0)}$ $\gap$ $\texttt{typ=2}$} $\isprefix\trace$ $\}$,  i.e. for any valid trace, if first two inputs are 1 and the third is 2, then it guarantees ``Read and Write access is granted'' $\texttt{RW}$.

\item[--] $\alphapropertytransform{\maximalsemantics{}}{\maximaltracepropertybot}$ $=\emptyset$, i.e. no valid trace can guarantee $\maximaltracepropertybot$. \qed
\end{itemize}
\end{example}
\vskip-1.5em

\vskip-.5em
\subsection{Observation of System Behaviors\label{subsec:monitor}}
Let $\maximalsemantics{}$ be the maximal trace semantics and $\latticeofmaximalproperties$  be the lattice of system behaviors designed as in Section \ref{subsec:lattice-trace-properties}.
Given any prefix trace $\trace\in\setofalltraces$, an observer can learn some information from it, more precisely, a maximal trace property $\traceproperty\in\latticeofmaximalproperties$ that is guaranteed by $\trace$ from the observer's perspective.
In this section, an \textit{observation} function $\observation{}{}{}$ is proposed to represent such a ``property learning process'' of the observer, which is formally defined in the following three steps.
\vskip-1em

\vskip-1em
\subsubsection{Inquiry Function\label{sec:inquiry-function}} First, an \textit{inquiry function} $\inquiryfunction{}{}$ is defined to map every trace $\trace\in\setofalltraces$ to the strongest maximal trace property in $\latticeofmaximalproperties$ that $\trace$ can guarantee.


\vskip-.3em
\begin{eqntabular}[fl]{@{}l}
\inquiryfunction{}{} \in \wp(\setofalltraces)\mapsto\wp(\wp(\setofalltraces))\mapsto\setofalltraces\mapsto\wp(\setofalltraces) 
\label{eq:inquiryfunction}\stepcounter{equation}\renumber{inquiry (\ref{eq:inquiryfunction})}\\
\inquiryfunction{}{\maximalsemantics{}, \latticeofmaximalproperties, \trace}\triangleq\nonumber\\
\quad let\;\alphapropertytransform{\semantics}{\traceproperty}=\{\trace\in\prefixesofset{\traceproperty}\mid\forall\trace'\in\semantics.\;\trace\isprefix\trace'\Rightarrow\trace'\in\traceproperty\}\;in\nonumber\\
\quad\quad\maximaltracepropertymeet\{\traceproperty\in\latticeofmaximalproperties\mid\trace\in\alphapropertytransform{\maximalsemantics{}}{\traceproperty}\} \nonumber
\end{eqntabular}
\vskip-.5em
Specially, for an invalid trace $\trace\not\in\prefixsemantics{}$, there does not exist any $\traceproperty\in\latticeofmaximalproperties$ such that $\trace\in\alphapropertytransform{\maximalsemantics{}}{\traceproperty}$, therefore $\inquiryfunction{}{\maximalsemantics{}, \latticeofmaximalproperties, \trace}$ $=\emptyset=\maximaltracepropertybot$.


\begin{corollary}\label{corollary:inquiry-guarantee} Given the semantics $\maximalsemantics{}$ and lattice $\latticeofmaximalproperties$ of system behaviors, if the inquiry function $\inquiryfunction{}{}$ maps a trace $\trace$ to a maximal trace property $\traceproperty\in\latticeofmaximalproperties$, then $\trace$ guarantees the satisfaction of $\traceproperty$ (i.e. every valid maximal trace that is greater than or equal to $\trace$ is guaranteed to have property $\traceproperty$).
\end{corollary}
\vskip-1em



\vskip-1em
\begin{lemma}\label{lemma:inquiry-decreasing} The inquiry function $\inquiryfunction{}{\maximalsemantics{}, \latticeofmaximalproperties}$ is decreasing on the inquired trace $\trace$: the greater (longer) $\trace$ is, the stronger property it can guarantee.
\end{lemma}

\begin{example}[Access Control, Continued]\label{ex:access-control-inquiry}
Using $\maximalsemantics{}$ defined in Example \ref{ex:access-control-semantics} and $\latticeofmaximalproperties$ defined in Example \ref{ex:access-control-lattice-maximal-trace-properties}, the inquiry function $\inquiryfunction{}{}$ of definition (\ref{eq:inquiryfunction}) is such that:

\begin{itemize}[leftmargin=*]
\item[--]$\inquiryfunction{}{\maximalsemantics{}, \latticeofmaximalproperties, {\footnotesize \texttt{apv=1}}}$ = $\maximaltracepropertytop$, i.e. {\footnotesize \texttt{apv=1}} can guarantee only $\maximaltracepropertytop$.

\item[--]$\inquiryfunction{}{\maximalsemantics{}, \latticeofmaximalproperties, {\footnotesize \texttt{apv=1} \gap \texttt{i1=0}}}$ = $\texttt{RF}$, i.e. after setting the first input as 0, ``Read Failure'' $\texttt{RF}$ is guaranteed.


\item[--]$\inquiryfunction{}{\maximalsemantics{}, \latticeofmaximalproperties, {\footnotesize \texttt{apv=1} \gap \texttt{i1=1}}}=\inquiryfunction{}{}(\maximalsemantics{}, \latticeofmaximalproperties, {\footnotesize \texttt{apv=1}}{\footnotesize \gap}{\footnotesize \texttt{i1=1}}$ ${\footnotesize \gap}{\footnotesize {\scriptstyle\neg}\texttt{(i1==0)}})=\maximaltracepropertytop$, i.e. if the first input is 1, only $\maximaltracepropertytop$ is guaranteed before entering the second input.

\item[--]$\inquiryfunction{}{\maximalsemantics{}, \latticeofmaximalproperties, {\footnotesize \texttt{apv=1} \gap \texttt{i1=1} \gap {\scriptstyle\neg}\texttt{(i1==0)} \gap \texttt{i2=0}}}=\texttt{RF}$, i.e. if the second input is 0, ``Read Failure'' $\texttt{RF}$ is guaranteed.

\item[--]$\inquiryfunction{}{\maximalsemantics{}, \latticeofmaximalproperties, {\footnotesize \texttt{apv=1} \gap \texttt{i1=1} \gap {\scriptstyle\neg}\texttt{(i1==0)} \gap \texttt{i2=1}}}=\texttt{RS}$, i.e. if first two inputs are 1, ``Read Success'' $\texttt{RS}$ is guaranteed.


\item[--]$\inquiryfunction{}{}(\maximalsemantics{}, \latticeofmaximalproperties, {\footnotesize \texttt{apv=1}}$ ${\footnotesize \gap}$ ${\footnotesize \texttt{i1=1}}$ ${\footnotesize \gap}$ ${\footnotesize {\scriptstyle\neg}\texttt{(i1==0)}}$ ${\footnotesize \gap}$ ${\footnotesize \texttt{i2=1}}$ ${\footnotesize \gap}$ ${\footnotesize {\scriptstyle\neg}\texttt{(i2==0)}}$ ${\footnotesize \gap}$ ${\footnotesize \texttt{typ=2}})$ = $\texttt{RW}$, i.e. if first two inputs are 1, after the third input is set to be 2, a stronger property ``Read and Write access is granted'' $\texttt{RW}$ is guaranteed. \qed
\end{itemize}
\end{example}
\vskip-1.5em

\vskip-1.5em
\subsubsection{Cognizance Function\label{sec:cognizance-function}} As discussed in (P3) of section \ref{subsec:example}, it is necessary to take the observer's cognizance into account.
Specifically, in program security, the observer's cognizance can be used to represent attackers' capabilities (e.g. what they can learn from the program execution).
Given a trace $\trace$ (not necessarily valid), if the observer cannot distinguish $\trace$ from some other traces, then he does not have an omniscient cognizance of $\trace$, and the \textit{cognizance} function $\cognizancefunction{}{\trace}$ is defined to include all traces indistinguishable from $\trace$.
\vskip-.5em
\begin{eqntabular}[fl]{@{}r@{~}c@{~}l}
\cognizancefunction{}{} &\in& \setofalltraces\mapsto\wp(\setofalltraces) \label{eq:cognizancefunction}\stepcounter{equation}\renumber{cognizance (\ref{eq:cognizancefunction})}\\
\cognizancefunction{}{\trace}&\triangleq& \{\trace'\in\setofalltraces\mid\text{\small observer\ cannot\ distinguish\ } \trace' \text{\small\ from\ } \trace\}\nonumber
\end{eqntabular}
\vskip-.4em
Such a cognizance function is extensive, i.e. $\forall\trace\in\setofalltraces.$ $\trace\in\cognizancefunction{}{\trace}$.  In particular, there is an \textit{omniscient observer} and its corresponding cognizance function is denoted as $\omniscientcognizance{}{}$ such that $\forall\trace\in\setofalltraces.\  \omniscientcognizance{}{\trace}=\{\trace\}$, which means that every trace is unambiguous to the omniscient observer.



To facilitate the proof of some desired properties for the observation function defined later, two assumptions are made here without loss of generality:

\begin{enumerate}[leftmargin=*,label={(A\arabic*)},itemsep=.5ex,topsep=.1em]

\item\label{assumption:cognizance1} The cognizance of a trace $\trace\trace'$ is the concatenation of cognizances of $\trace$ and $\trace'$. I.e. $\forall\trace,\trace'\in\setofalltraces.$ $\cognizancefunction{}{\trace\trace'}$ = $\{\tau\tau'\mid\tau\in\cognizancefunction{}{\trace}\wedge\tau'\in\cognizancefunction{}{\trace'}\}$.

\item\label{assumption:cognizance2} Given an invalid trace, the cognizance function would not return a valid trace. I.e. $\forall\trace\in\setofalltraces.$ $ \trace\not\in\prefixsemantics{}$ $\Rightarrow$ $\cognizancefunction{}{\trace}\cap\prefixsemantics{}$ = $\emptyset$.
\end{enumerate}

To make the assumption \ref{assumption:cognizance1} sound, we must have $\cognizancefunction{}{\emptytrace}=\{\emptytrace\}$, because otherwise, for any non-empty trace $\trace$, $\cognizancefunction{}{\trace}=$ $\cognizancefunction{}{\trace\emptytrace}$ = $\{\tau\tau'\mid\tau\in\cognizancefunction{}{\trace}\wedge\tau'\in\cognizancefunction{}{\emptytrace}\}$ does not have a fixpoint. 
In practice, $\{\pair{\trace}{\trace'}\mid\trace'\in\cognizancefunction{}{\trace}\}$ is an equivalence relation, but the symmetry and transitivity property are not used in the proofs.
\vskip-.5em

\vskip-.5em
\begin{example}[Access Control, Continued]\label{ex:access-control-cognizance}
Consider two separate observers.

(i) For an omniscient observer: $\forall\trace\in\setofalltraces.$ $\omniscientcognizance{}{\trace}=\{\trace\}$.

(ii) For an observer representing the second admin who is unaware of the first input: $\cognizancefunction{}{{\footnotesize \texttt{i1=0}\gap \texttt{i1==0} \gap \texttt{apv=0}}}=\cognizancefunction{}{{\footnotesize \texttt{i1=1}\gap {\scriptstyle\neg}\texttt{(i1==0)}}}=$ \{{\footnotesize $\texttt{i1=0}$ $\gap$ $\texttt{i1==0}$ $\gap$ $\texttt{apv=0}$}, {\footnotesize $\texttt{i1=1}$ $\gap$ ${\scriptstyle\neg}\texttt{(i1==0)}$}\}, i.e. this observer cannot distinguish whether the first input is 0 or 1. Thus, for a prefix trace in which the first two inputs are 0, $\cognizancefunction{}{{\footnotesize \texttt{apv=1} \gap \texttt{i1=0} \gap \texttt{i1==0} \gap \texttt{apv=0} \gap \texttt{i2=0}}}$ = \{{\footnotesize $\texttt{apv=1}$ $\gap$ $\texttt{i1=0}$ $\gap$ $\texttt{i1==0}$ $\gap$ $\texttt{apv=0}$ $\gap$ $\texttt{i2=0}$}, {\footnotesize $\texttt{apv=1}$ $\gap$ $\texttt{i1=1}$ $\gap$ ${\scriptstyle\neg}\texttt{(i1==0)}$ $\gap$ $\texttt{i2=0}$}\}, where {\footnotesize $\texttt{apv=1}$} and {\footnotesize $\texttt{i2=0}$} are known by this observer. In the same way, its cognizance on other traces can be generated.
\qed\end{example}
\vskip-1.5em

\vskip-1.5em
\subsubsection{Observation Function\label{sec:observation-function}} 
For an observer with cognizance function $\cognizancefunction{}{}$, given a single trace $\trace$, the observer cannot distinguish $\trace$ with traces in $\cognizancefunction{}{\trace}$. In order to formalize the information that the observer can learn from $\trace$, we apply the inquiry function $\inquiryfunction{}{}$ on each trace in $\cognizancefunction{}{\trace}$, and get a set of maximal trace properties. By joining them together, we get the strongest property in $\latticeofmaximalproperties$ that $\trace$ can guarantee from the observer's perspective. Such a process is defined as the \textit{observation} function $\observation{}{}{\maximalsemantics{}, \latticeofmaximalproperties, \cognizancefunction{}{}, \trace}$.
\vskip-.5em
\begin{eqntabular}[fl]{@{}l}
\observation{}{}{}\in\wp(\setofalltraces)\mapsto\wp(\wp(\setofalltraces))\mapsto(\setofalltraces\mapsto\wp(\setofalltraces))\mapsto\setofalltraces\mapsto\wp(\setofalltraces) 
\nonumber\\
\observation{}{}{\maximalsemantics{}, \latticeofmaximalproperties, \cognizancefunction{}{}, \trace}\triangleq 
\label{eq:observation}\stepcounter{equation}
\renumber{observation (\ref{eq:observation})}\\
\quad let\;\alphapropertytransform{\semantics}{\traceproperty}=\{\trace\in\prefixesofset{\traceproperty}\mid\forall\trace'\in\semantics.\;\trace\isprefix\trace'\Rightarrow\trace'\in\traceproperty\}\;in\nonumber\\
\quad\quad let\;\inquiryfunction{}{\semantics, \latticeofproperties,\trace} =  \maximaltracepropertymeet\{\traceproperty\in\latticeofproperties\mid\trace\in\alphapropertytransform{\semantics}{\traceproperty}\}\;in\nonumber
\\
\quad\quad\quad\maximaltracepropertyjoin\{\inquiryfunction{}{\maximalsemantics{}, \latticeofmaximalproperties, \trace'}\mid\trace'\in\cognizancefunction{}{\trace}\}. \nonumber
\end{eqntabular}

From the above definition, it is easy to see that, for every invalid trace $\trace$, $\observation{}{}{\maximalsemantics{}, \latticeofmaximalproperties, \cognizancefunction{}{}, \trace}=\maximaltracepropertybot$, since every trace $\trace'$ in $\cognizancefunction{}{\trace}$ is invalid by \ref{assumption:cognizance2} and $\inquiryfunction{}{\maximalsemantics{}, \latticeofmaximalproperties, \trace'}=\maximaltracepropertybot$. In addition, for an omniscient observer with cognizance function $\omniscientcognizance{}{}$, its observation $\observation{}{}{\maximalsemantics{}, \latticeofmaximalproperties, \omniscientcognizance{}{}, \trace}$ $=$ $\inquiryfunction{}{\maximalsemantics{}, \latticeofmaximalproperties, \trace}$.

\begin{corollary}\label{corollary:observation-guarantee} For any observer with cognizance $\cognizancefunction{}{}$, if the corresponding observation function maps a trace $\trace$ to a maximal trace property $\traceproperty\in\latticeofmaximalproperties$, then $\trace$ guarantees the satisfaction of property $\traceproperty$ (i.e. every valid maximal trace that is greater than or equal to $\trace$ is guaranteed to have property $\traceproperty$).
\end{corollary}


\begin{lemma}\label{lemma:observation-decreasing} The observation function $\observation{}{}{\maximalsemantics{}, \latticeofmaximalproperties, \cognizancefunction{}{}}$ is decreasing on the observed trace $\trace$: the greater (longer) $\trace$ is, the stronger property it can observe.
\end{lemma}

\begin{example}[Access Control, Continued]\label{ex:access-control-observation}
For an omniscient observer, the observation function is identical to the inquire function in Example \ref{ex:access-control-inquiry}. 
If the cognizance of the second admin defined in Example \ref{ex:access-control-cognizance} is adopted, we get an observation function that works exactly the same as the dashed arrows in Fig.\ref{fig:analysis-framework}:
\vskip-.5em

\vskip-.5em
\begin{itemize}[leftmargin=*]
\item[--]$\observation{}{}{\maximalsemantics{}, \latticeofmaximalproperties, \cognizancefunction{}{}, {\footnotesize \texttt{apv=1} \gap \texttt{i1=0}}}$ 
$=$ $\inquiryfunction{}{}$ $(\maximalsemantics{}, \latticeofmaximalproperties,$ ${\footnotesize \texttt{apv=1} \gap}$  ${\footnotesize \texttt{i1=0}})$ 
$\maximaltracepropertyjoin$ $\inquiryfunction{}{}(\maximalsemantics{}, \latticeofmaximalproperties,$ ${\footnotesize \texttt{apv=1} \gap \texttt{i1=1}})$ 
$=$ $\texttt{RF}$ $\maximaltracepropertyjoin$ $\maximaltracepropertytop$ = $\maximaltracepropertytop$, i.e. even if the first input is already 0 in the trace, no property except $\maximaltracepropertytop$ can be guaranteed for the second admin.

\item[--]$\observation{}{}{\maximalsemantics{}, \latticeofmaximalproperties, \cognizancefunction{}{}, {\footnotesize \texttt{apv=1} \gap \texttt{i1=0} \gap \texttt{i1==0} \gap \texttt{apv=0} \gap \texttt{i2=1}}}$ $=$
$\inquiryfunction{}{}$ $(\maximalsemantics{},$ $\latticeofmaximalproperties,$ ${\footnotesize \texttt{apv=1} \gap}$ ${\footnotesize \texttt{i1=0} \gap \texttt{i1==0} \gap \texttt{apv=0} \gap \texttt{i2=1}})$ $\maximaltracepropertyjoin$
$\inquiryfunction{}{}$ $(\maximalsemantics{},$ $\latticeofmaximalproperties,$ ${\footnotesize \texttt{apv=1} \gap \texttt{i1=1} \gap {\scriptstyle\neg}\texttt{(i1==0)} \gap \texttt{i2=1}})$ $=$ $\texttt{RF}$ $\maximaltracepropertyjoin$ $\maximaltracepropertytop$ $=$ $\maximaltracepropertytop$, i.e. if the second input is 1, only $\maximaltracepropertytop$ can be guaranteed.

\item[--]$\observation{}{}{\maximalsemantics{}, \latticeofmaximalproperties, \cognizancefunction{}{}, {\footnotesize \texttt{apv=1} \gap \texttt{i1=0} \gap \texttt{i1==0} \gap \texttt{apv=0} \gap \texttt{i2=0}}}$ $=$
$\inquiryfunction{}{}$ $(\maximalsemantics{},$ $\latticeofmaximalproperties,$ ${\footnotesize \texttt{apv=1} \gap}$ ${\footnotesize \texttt{i1=0} \gap \texttt{i1==0} \gap \texttt{apv=0} \gap \texttt{i2=0}})$ $\maximaltracepropertyjoin$
$\inquiryfunction{}{}$ $(\maximalsemantics{},$ $\latticeofmaximalproperties,$ ${\footnotesize \texttt{apv=1} \gap \texttt{i1=1} \gap {\scriptstyle\neg}\texttt{(i1==0)} \gap \texttt{i2=0}})$ $=$ $\texttt{RF}\maximaltracepropertyjoin\texttt{RF}$ $=\texttt{RF}$, i.e. $\texttt{RF}$ is guaranteed only after the second input is entered 0. \qed
\end{itemize}
\end{example}
\vskip-1em

\vskip-1em
\subsection{Formal Definition of Responsibility\label{subsec:simple-definition}}

Using the three components of responsibility analysis 
introduced above, responsibility
is formally defined as the \textit{responsibility abstraction} $\alpharesponsibility{}{}{}$ in (\ref{eq:alpharesponsibility}). Specifically, the first parameter is the maximal trace semantics $\maximalsemantics{}$, the second parameter is the lattice $\latticeofmaximalproperties$ of system behaviors, the third parameter is the cognizance function of a given observer, the fourth parameter is the behavior $\behavior$ whose responsibility is of interest, and the last parameter is the analyzed traces $\tracepropertyT$. 

Consider every trace $\tracehistory\traceresponsible\tracefuture\in\tracepropertyT$ where H, R and F respectively stand for \textit{History}, \textit{Responsible part} and \textit{Future}. If $\emptyset\subsetneq\observation{}{}{\maximalsemantics{}, \latticeofmaximalproperties, \cognizancefunction{}{}, \tracehistory\traceresponsible}\subseteq\behavior\subsetneq\observation{}{}{\maximalsemantics{}, \latticeofmaximalproperties, \cognizancefunction{}{}, \tracehistory}$ holds, then $\tracehistory$ does not guarantee the behavior $\behavior$, while $\tracehistory\traceresponsible$ guarantees a behavior which is at least as strong as $\behavior$ and is not the invalidity property represented by $\maximaltracepropertybot=\emptyset$. Therefore, $\traceresponsible$ is said to be \textit{responsible} for ensuring behavior $\behavior$ in the trace $\tracehistory\traceresponsible\tracefuture$.

In particular, the length of $\traceresponsible$ is restricted to be 1 (i.e. $|\traceresponsible|=1$), such that the responsible entity $\traceresponsible$ must be a single event and the responsibility analysis could be as refined as possible. Otherwise, if we do not have such a restriction, then for every analyzed trace $\trace\in\tracepropertyT$ where the behavior $\behavior$ holds, the responsibility analysis may split the trace $\trace$ into three parts $\trace=\tracehistory\traceresponsible\tracefuture$ such that $\tracehistory=\emptytrace$, $\traceresponsible=\trace$ and $\tracefuture=\emptytrace$. In such a case, $\emptyset\subsetneq\observation{}{}{\maximalsemantics{}, \latticeofmaximalproperties, \cognizancefunction{}{}, \tracehistory\traceresponsible}\subseteq\behavior\subsetneq\observation{}{}{\maximalsemantics{}, \latticeofmaximalproperties, \cognizancefunction{}{}, \tracehistory}$ holds, and the whole trace $\trace$ would be found responsible for $\behavior$. This result is trivially correct, but too coarse to be useful in practice.

\begin{mdframed}[innertopmargin=0pt,innerbottommargin=6pt,frametitle={\colorbox{white}{\normalfont\textit{Responsibility Abstraction} $\alpharesponsibility{}{}{}$}},frametitleaboveskip=-7pt]\abovedisplayskip0pt

\begin{eqntabular}[fl]{@{}l@{}}
\hskip-.8em
\alpharesponsibility{}{}{} \in \wp(\setofalltraces)\mapsto\wp(\wp(\setofalltraces))\mapsto(\setofalltraces\mapsto\wp(\setofalltraces)) \nonumber
\\
\qquad\mapsto\wp(\setofalltraces)\mapsto\wp(\setofalltraces)\mapsto\wp(\setofalltraces\!\times\setofallevents\times\setofalltraces) \label{eq:alpharesponsibility}\stepcounter{equation}\renumber{(\ref{eq:alpharesponsibility})~~}
\\
\hskip-.8em
\alpharesponsibility{}{}{\maximalsemantics{}, \latticeofmaximalproperties, \cognizancefunction{}{}, \behavior, \tracepropertyT}\triangleq \nonumber
\\
let\;\alphapropertytransform{\semantics}{\traceproperty}=\{\trace\in\prefixesofset{\traceproperty}\mid\forall\trace'\in\semantics.\;\trace\isprefix\trace'\Rightarrow\trace'\in\traceproperty\}\;in\nonumber
\\
\quad let\;\inquiryfunction{}{\semantics, \latticeofproperties,\trace} =  \maximaltracepropertymeet\{\traceproperty\in\latticeofproperties\mid\trace\in\alphapropertytransform{\semantics}{\traceproperty}\}\;in\nonumber
\\
\quad\quad let\; \observation{}{}{\semantics, \latticeofproperties, \cognizancefunction{}{}, \trace} = \maximaltracepropertyjoin\{\inquiryfunction{}{\semantics, \latticeofproperties, \trace'}\mid\trace'\in\cognizancefunction{}{\trace}\}\;in\nonumber
\\
\quad\quad\quad\{\triple{\tracehistory}{\traceresponsible}{\tracefuture}\mid\tracehistory\traceresponsible\tracefuture\in\tracepropertyT\wedge|\traceresponsible|=1\wedge{}\nonumber
\\
\quad\quad\quad\quad
\emptyset\,{\subsetneq}\,\observation{}{}{\maximalsemantics{}\!, \latticeofmaximalproperties\!, \cognizancefunction{}{}, \tracehistory\traceresponsible}\subseteq\behavior\subsetneq\observation{}{}{\maximalsemantics{}\!, \latticeofmaximalproperties\!, \cognizancefunction{}{}, \tracehistory}\}\nonumber
\end{eqntabular}
\end{mdframed}

Since $\alpharesponsibility{}{}{\maximalsemantics{}, \latticeofmaximalproperties, \cognizancefunction{}{}, \behavior}$ preserves joins on analyzed traces $\tracepropertyT$, we have a Galois connection:
$\pair{\wp(\setofalltraces)}{\subseteq}\galois{\alpharesponsibility{}{}{\maximalsemantics{},\,\latticeofmaximalproperties,\,\cognizancefunction{}{},\,\behavior}}{\gamma_R(\maximalsemantics{},\,\latticeofmaximalproperties,\,\cognizancefunction{}{},\,\behavior)}\pair{\wp(\setofalltraces\!\times\setofallevents\times\setofalltraces)}{\subseteq}$.

\begin{lemma}\label{lemma:responsibility-definition}  If $\traceresponsible$ is said to be responsible for a behavior $\behavior$ in a valid trace $\tracehistory\traceresponsible\tracefuture$, then $\tracehistory\traceresponsible$ guarantees the occurrence of behavior $\behavior$, and there must exist another valid prefix trace $\tracehistory\traceresponsible'$ such that the behavior $\behavior$ is not guaranteed.
\end{lemma}

Recall the three desired points (time, free choices and cognizance) for defining responsibility in section \ref{subsec:example}. It is obvious that 
$\alpharesponsibility{}{}{}$ has taken both the temporal ordering of events and the observer's cognizance into account. As for the free choices, it is easy to find from lemma \ref{lemma:responsibility-definition} that, if $\traceresponsible$ is determined by its history trace $\tracehistory$ and is not free to make choices (i.e. $\forall\tracehistory\traceresponsible,\tracehistory\traceresponsible'\in\prefixsemantics{}.$ $\traceresponsible=\traceresponsible'$), then $\traceresponsible$ cannot be responsible for any behavior in the trace $\tracehistory\traceresponsible\tracefuture$.

\subsection{Responsibility Analysis \label{subsec:responsibility-analysis}}

To sum up, the responsibility analysis typically consists of four steps: 
\textbf{I)} collect the system's trace semantics $\maximalsemantics{}$ (in Section \ref{sec:System-Semantics}); \textbf{II)} build the complete lattice of maximal trace properties of interest $\latticeofmaximalproperties$ (in Section \ref{sec:Lattice-of-Maximal-Trace-Properties}); 
\textbf{III)} derive an inquiry function $\inquiryfunction{}{}$ from $\latticeofmaximalproperties$, define a cognizance function $\cognizancefunction{}{}$ for each observer, and create the corresponding observation function $\observation{}{}{}$ (in Section \ref{sec:observation-function});
\textbf{IV)} specify the behavior $\behavior$ of interest and the analyzed traces $\tracepropertyT$, and apply the responsibility abstraction $\alpharesponsibility{}{}{\maximalsemantics{}, \latticeofmaximalproperties, \cognizancefunction{}{}, \behavior, \tracepropertyT}$ to get the analysis result (in Section \ref{subsec:simple-definition}). Hence, the responsibility analysis is essentially an abstract interpretation of the event trace semantics. 

In the above definition of responsibility, the semantics and lattice of system behaviors are concrete, and they are explicitly displayed in the access control example for the sake of clarity. However, they may be uncomputable in practice, and we do not require programmers to provide them in the implementation of responsibility analysis. Instead, they are provided in the abstract, using an abstract interpretation-based static analysis that is sketched
\iflong
in the appendix \ref{sec:abstract-responsibility}.
\else
in \cite{SAS-2019-extended-version}.
\fi

\begin{example}[Access Control, Continued]\label{ex:access-control-responsibility-formal}
Using the observation functions created in example \ref{ex:access-control-observation}, the abstraction $\alpharesponsibility{}{}{}$ can analyze the responsibility of a certain behavior $\behavior$ in the set $\tracepropertyT$ of traces. Suppose we want to analyze ``Read Failure'' in every possible execution, then $\behavior$ is $\texttt{RF}$, and $\tracepropertyT$ includes all valid maximal traces, i.e. $\tracepropertyT=\maximalsemantics{}$. Thus, $\alpharesponsibility{}{}{\maximalsemantics{}, \latticeofmaximalproperties, \cognizancefunction{}{}, \texttt{RF}, \maximalsemantics{}}$ computes the responsibility analysis result, which is essential the same as desired in Example \ref{ex:access-control-responsibility}.

Furthermore, the responsibility of ``granting write access'' can be analyzed by setting the behavior $\behavior$ as $\texttt{RW}$ instead, and we get the following result. To the cognizance of an omniscient observer, in every execution that both the first two inputs are 1, the third input (i.e. system setting of permission type) is responsible for $\texttt{RW}$. Meanwhile, to the cognizance of the second admin who is unaware of the first input, no one is found responsible for $\texttt{RW}$, because whether the write access fails or not is always uncertain, from the second admin's perspective.
\qed
\end{example}
\section{Examples of Responsibility Analysis\label{sec:application}}

\vskip-.5em
Responsibility is a broad concept, and our definition of responsibility based on the abstraction of event trace semantics is universally applicable in various scientific fields. We have examined every example supplied in actual cause \cite{Halpern-Pearl-2001,Halpern-Pearl-2005} and found that our definition of responsibility can handle them well, in which actions like ``drop a lit match in the forest'' or ``throw a rock at the bottle'' are treated as events in the trace. In the following, we will illustrate the responsibility analysis by two more examples: the ``negative balance" problem of a withdrawal transaction, and the information leakage problem.

\subsection{Responsibility Analysis of ``Negative Balance'' Problem}

\begin{example}[Negative Balance]\label{ex:negative-balance}
Consider the withdrawal transaction program in Fig. \ref{fig:negative-balance} in which the ${\footnotesize \texttt{query\_database()}}$ function gets the balance of a certain bank account before the transaction, and ${\footnotesize \texttt{input()}}$ specifies the withdrawal amount that is positive. When the withdrawal transaction completes, if the balance is negative, which entity in the program shall be responsible for it?
\qed
\end{example}

It is not hard to see that, the ``negative balance" problem can be transformed into an equivalent buffer overflow problem, where the memory of size ${\footnotesize \texttt{balance}}$ is allocated and the index at  ${\footnotesize \texttt{n-1}}$ is visited. Although this problem has been well studied, it suffices to demonstrate the advantages of responsibility analysis over dependency/causality analysis.

\vskip-2em
\begin{figure}[h]
\begin{minipage}{.55\textwidth}
\footnotesize\ttfamily
\vskip.4em
\begin{lstlisting}
1: balance = query_database();
2: n = input(); //Positive
3: balance -= n;
\end{lstlisting}
\vskip-1em
\caption{Withdrawal Transaction Program\label{fig:negative-balance}}
\end{minipage}
\begin{minipage}{.45\textwidth}
\centering
\begin{tikzpicture}[scale=.6]
\footnotesize
     \draw (1,0) node[name=BOT] {{$\maximaltracepropertybot\rlap{ = $\emptyset$}$}};
     \draw (0,1) node[name=O] {\color{black}{\texttt{NB}}};
     \draw (2,1) node[name=NO] {\color{black}{$\neg\texttt{NB}$}};
     \draw (1,2) node[name=TOP] {{$\maximaltracepropertytop\rlap{ = $ \maximalsemantics{}$}$}};
     \draw[semithick] (BOT) -- (O);
     \draw[semithick] (BOT) -- (NO);
     \draw[semithick] (O) -- (TOP);
     \draw[semithick] (NO) -- (TOP);
\end{tikzpicture}
\vskip-1em
\caption{Lattice of Behaviors\label{fig:lattice-negative-balance-maximal}}
\end{minipage}
\end{figure}
\vskip-1.5em

As discussed in section \ref{subsec:responsibility-analysis}, the responsibility analysis consists of four steps.
For the sake of simplicity, we consider only the omniscient observer here.

\begin{enumerate}[leftmargin=*,label=(\arabic*)]
\item Taking each assignment as an event, each maximal trace in this program is of length 3, and the program's maximal trace semantics consists of infinite number of such traces. E.g. {\footnotesize $\texttt{balance=0}$ $\gap$ $\texttt{n=5}$ $\gap$ $\texttt{balance=-5}$} denotes a maximal execution, in which the balance before the transaction is 0 and the withdrawal amount is 5 such that ``negative balance" occurs. 

\item Since ``negative balance" is the only behavior that we are interested here, a lattice $\latticeofmaximalproperties$ of maximal trace properties in Fig. \ref{fig:lattice-negative-balance-maximal} with four elements can be built, where $\texttt{NB}$ (Negative Balance) is the set of maximal traces where  the value of ${\footnotesize \texttt{balance}}$ is negative at the end, and $\neg\texttt{NB}$ is its complement. 

%
%

\item Using the omniscient observer's cognizance $\omniscientcognizance{}{}$, the observation function $\observation{}{}{}$ can be easily derived from the lattice $\latticeofmaximalproperties$, such that: 

\begin{itemize}[leftmargin=*]
\item[--] $\observation{}{}{\maximalsemantics{}, \latticeofmaximalproperties, \omniscientcognizance{}{}, \emptytrace}$ $=\maximaltracepropertytop$;

\item[--] $\observation{}{}{\maximalsemantics{}, \latticeofmaximalproperties, \omniscientcognizance{}{}, \footnotesize \texttt{balance=i}}$ $=\texttt{NB}$ where ${\footnotesize \texttt{i}}\leq 0$, i.e. if the balance before the transaction is negative or 0, the occurrence of ``negative balance" is guaranteed before the withdrawal amount ${\footnotesize \texttt{n}}$ is entered;

\item[--] $\observation{}{}{\maximalsemantics{}, \latticeofmaximalproperties, \omniscientcognizance{}{}, \footnotesize \texttt{balance=i}}$ $=\maximaltracepropertytop$ where ${\footnotesize \texttt{i}}> 0$, i.e. if the balance before the transaction is strictly greater than 0, whether ``negative balance" occurs or not has not been decided;

\item[--] $\observation{}{}{\maximalsemantics{}, \latticeofmaximalproperties, \omniscientcognizance{}{}, \footnotesize \texttt{balance=i} \gap \texttt{n=j}}=\texttt{NB}$ where ${\footnotesize \texttt{i}}> 0$ and {\footnotesize $\texttt{j}> \texttt{i}$}, i.e. ``negative balance" is guaranteed to occur immediately after ${\footnotesize \texttt{input()}}$ returns a value strictly greater than {\footnotesize $\texttt{balance}$}; 

\item[--] $\observation{}{}{\maximalsemantics{}, \latticeofmaximalproperties, \omniscientcognizance{}{}, \footnotesize \texttt{balance=i} \gap \texttt{n=j}}=\neg\texttt{NB}$ where ${\footnotesize \texttt{i}}> 0$ and {\footnotesize $\texttt{j}\leq \texttt{i}$}, i.e. ``negative balance" is guaranteed not to occur immediately after ${\footnotesize \texttt{input()}}$ returns a value less than or equal to {\footnotesize $\texttt{balance}$}. 
\end{itemize}

\item Suppose the behavior $\behavior=\texttt{NB}$ and the analyzed traces $\tracepropertyT$ $=\maximalsemantics{}$, the abstraction $\alpharesponsibility{}{}{\maximalsemantics{}, \latticeofmaximalproperties, \omniscientcognizance{}{}, \behavior, \tracepropertyT}$ gets the following result. If ${\footnotesize \texttt{query\_database()}}$ returns 0 or a negative value, no matter what value ${\footnotesize \texttt{input()}}$ returns, the function ${\footnotesize \texttt{query\_database()}}$ (i.e. event $\footnotesize \texttt{balance=i}$) is responsible for ``negative balance", and further responsibility analysis shall be applied on the previous transactions of the database. Otherwise, if ${\footnotesize \texttt{query\_database()}}$ returns a value strictly greater than 0, the function ${\footnotesize \texttt{input()}}$ (i.e. event $\footnotesize \texttt{n=j}$) takes the responsibility for ``negative balance", thus ``negative balance" can be prevented by configuring the permission granted to ${\footnotesize \texttt{input()}}$ such that its permitted return value must be less than or equal to the returned value of ${\footnotesize \texttt{query\_database()}}$.
\end{enumerate}



\subsection{Responsibility Analysis of Information Leakage}


Essentially, responsibility analysis of information leakage is the same as read failure or ``negative balance" problem, and the only significant distinction is on defining the behaviors of interest. 
Here we adopt the notion of non-interference \cite{DBLP:conf/sp/GoguenM82a} to represent the behavior of information leakage. 

In the program, the inputs and outputs are classified as either \textit{Low} (public, low sensitivity) or \textit{High} (private, high sensitivity). For a given trace $\trace$, if there is another trace $\trace'$ such that they have the same low inputs but different low outputs, then the trace $\trace$ is said to leak private information. If no trace in the program leaks private information (i.e. every two traces with the same low inputs have the same low outputs, regardless of the high inputs), the program is secure and has the non-interference property. Thus, for any program with maximal trace semantics $\maximalsemantics{}$, the behavior of ``Information Leakage'' $\texttt{IL}$ is represented as the set of leaky traces, i.e. $\texttt{IL}$ $=$ $\{\trace\in\maximalsemantics{}\mid$ $\exists\trace'\in\maximalsemantics{}.$ $low\_inputs(\trace)$ $=$ $low\_inputs(\trace')$ $\wedge$ $low\_outputs(\trace)$ $\neq$ $low\_outputs(\trace')\}$, where functions $low\_inputs$ and $low\_outputs$ collects low inputs and outputs along the trace, respectively. The behavior of ``No information Leakage'' $\texttt{NL}$ is the complement of $\texttt{IL}$, i.e. $\texttt{NL}$ $=$ $\{\trace\in\maximalsemantics{}\mid$ $\forall\trace'\in\maximalsemantics{}.$ $low\_inputs(\trace)$ $=$ $low\_inputs(\trace')$ $\Rightarrow$ $low\_outputs(\trace)$ $=$ $low\_outputs(\trace')\}$. Thus, the lattice $\latticeofmaximalproperties$ of maximal trace properties regarding information leakage can be built as in in Fig. \ref{fig:lattice-information-leakage-maximal}. Further, the corresponding observation function $\observation{}{}{}$ can be created, and the analysis result can be obtained by applying the responsibility abstraction. 

\vskip-2em
\begin{center}
\begin{figure}[h]
\vskip-1.6em
\hskip2.0em
\begin{subfigure}[t]{.5\textwidth}
\centering
\begin{tikzpicture}[scale=.7]
\footnotesize
     \draw (1,0) node[name=BOT] {{$\maximaltracepropertybot\rlap{ = $\emptyset$}$}};
     \draw (0,1) node[name=IL] {\color{black}{\texttt{IL}}};
     \draw (2,1) node[name=NL] {\color{black}{\texttt{NL}}};
     \draw (1,2) node[name=TOP] {{$\maximaltracepropertytop\rlap{ = $ \maximalsemantics{}$}$}};

     \draw[semithick] (BOT) -- (IL);
     \draw[semithick] (BOT) -- (NL);
     \draw[semithick] (IL) -- (TOP);
     \draw[semithick] (NL) -- (TOP);

\end{tikzpicture}
\end{subfigure}
\begin{subfigure}[b]{.2\textwidth}\footnotesize
\belowdisplayskip.8em
\begin{eqntabular*}{@{}L@{~}C@{~}L@{}}
\rlap{Behaviors:}&&\\
\texttt{IL}&:&Information Leakage\\
\texttt{NL}&:&No information Leakage
\end{eqntabular*}%
\mbox{}
\end{subfigure}%
 \caption{Lattice of Behaviors regarding Information Leakage}
\label{fig:lattice-information-leakage-maximal}
\end{figure}
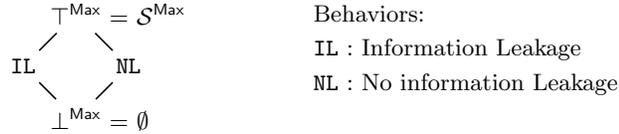
\vskip-2.5em
\end{center}

\vskip-1em
Notice that we are interested in analyzing only the insecure programs in which some traces leak private information while others do not, i.e. $\texttt{IL}\subsetneq\maximaltracepropertytop$. For the erroneous programs where every trace leaks private information, i.e. $\texttt{IL}=\maximaltracepropertytop$, we need to admit that our responsibility analysis cannot identify any entity responsible for the leakage, unless ``launching the program'' is treated as an event and it would be found responsible for leaking private information. 

\section{Related Work\label{sec:related-work}}

\subsubsection{Definition of Causality and Responsibility}

Hume \cite{Hume-1748} is the first one to specify causation by counterfactual dependence \cite{sep-causation-counterfactual}. The best known counterfactual theory of causation is proposed by Lewis \cite{Lewis-causation-1973}, which defines causation as a transitive closure of counterfactual dependencies. 
Halpern and Pearl \cite{Halpern-Pearl-2001,Halpern-Pearl-2005,Pearl-Causality-2012} defines actual causality based on SEM and extends counterfactual dependency to allow ``contingent dependency''. Chockler and Halpern \cite{DBLP:journals/jair/ChocklerH04} defines responsibility to have a quantitative measure of the relevance between causes and effects, and defines blame to consider the epistemic state of an agent. Their application of actual causality, responsibility and blame is mainly on artificial intelligence.

Our definition of responsibility also adopts the idea of counterfactual dependence in the sense that, suppose an event $\traceresponsible$ is said to be responsible for behavior $\behavior$ in the trace $\tracehistory\traceresponsible$, there must exist another event $\traceresponsible'$ such that, if $\traceresponsible$ is replaced by $\traceresponsible'$, then $\behavior$ is not guaranteed (by lemma \ref{lemma:responsibility-definition}). 
\vskip-1em

\vskip-1em
\subsubsection{Error Cause Localization}
Classic program analysis techniques, e.g. dependency analysis \cite{DBLP:conf/popl/AbadiBHR99,DBLP:journals/mscs/CheneyAA11,DBLP:conf/esop/UrbanM18} and program slicing \cite{DBLP:journals/tse/Weiser84,Weiser81,DBLP:journals/infsof/KorelR98,DBLP:conf/pldi/AgrawalH90}, are useful in detecting the code that may be relevant to errors, but fail to localize the cause of error. 

In recent years, there are many papers \cite{DBLP:conf/popl/BallNR03,DBLP:journals/sttt/GroceCKS06,DBLP:conf/pldi/JoseM11,DBLP:journals/entcs/GriesmayerSB07,DBLP:conf/kbse/RenierisR03,DBLP:conf/tacas/RaviS04,DBLP:conf/sigsoft/QiRLV09,DBLP:conf/tacas/JinRS02} on fault localization for counterexample traces, and most of them compare multiple traces produced by a model checker and build a heuristic metric to localize the point from which error traces separate from correct traces. 
Other related papers include error diagnosis by abductive/backward inference \cite{DBLP:conf/pldi/DilligDA12}, tracking down bugs by dynamic invariant detection \cite{DBLP:conf/icse/HangalL02}. Actual causality is applied to explain counterexamples from model checker \cite{DBLP:journals/fmsd/BeerBCOT12} and estimate the coverage of specification \cite{DBLP:journals/tocl/ChocklerHK08}.
Besides, there are researches on analyzing causes of specific security issues. E.g. King et al. \cite{DBLP:conf/sigsoft/KingJJS08} employ a blame dependency graph to explain the source of information flow violation and generate a program slice as the error report. 


Compared to the above techniques, this paper succeeds to formally define the cause or responsibility, and the proposed responsibility analysis, which does not require a counterexample from the model checker, is sound, scalable  and generic to cope with various problems.

\section{Conclusion and Future Work}


This paper formally defines responsibility as an abstraction of event trace semantics. Typically, the responsibility analysis consists of four steps: collect the trace semantics, build a lattice of behaviors of interest, create an observation function for each observer, and apply the responsibility abstraction on analyzed traces. Its effectiveness has been demonstrated by several examples.

In the future, we intent to: (1) formalize the abstract responsibility analysis that is sketched 
\iflong
in the appendix \ref{sec:abstract-responsibility},
\else
in \cite{SAS-2019-extended-version},
\fi
(2) build a lattice of responsibility abstractions to cope with possible alternative weaker or stronger definitions of responsibility,
(3) generalize the definition of cognizance function as an abstraction of system semantics, and 
(4) study the responsibility analysis of probabilistic programs.
\section*{Acknowledgment}
 This work was supported in part by NSF Grant 
  CNS-1446511.  Any opinions, findings, and
  conclusions or recommendations expressed in this material are those
  of the authors and do not necessarily reflect the views of the
  National Science Foundation. P. Cousot thanks Marco Pistoia for initial discussions on responsibility while visiting the Thomas J. Watson Research Center at Hawthorne in 2005.

\begin{CJK*}{UTF8}{gbsn}
\bibliography{responsibility}

\begin{thebibliography}{10}
\providecommand{\url}[1]{\texttt{#1}}
\providecommand{\urlprefix}{URL }
\providecommand{\doi}[1]{https://doi.org/#1}

\bibitem{DBLP:conf/popl/AbadiBHR99}
Abadi, M., Banerjee, A., Heintze, N., Riecke, J.G.: A core calculus of
  dependency. In: {POPL}. pp. 147--160. {ACM} (1999)

\bibitem{DBLP:conf/pldi/AgrawalH90}
Agrawal, H., Horgan, J.R.: Dynamic program slicing. In: {PLDI}. pp. 246--256.
  {ACM} (1990)

\bibitem{DBLP:journals/pacmpl/AguirreBG0S17}
Aguirre, A., Barthe, G., Gaboardi, M., Garg, D., Strub, P.: A relational logic
  for higher-order programs. {PACMPL}  \textbf{1}({ICFP}),  21:1--21:29 (2017)

\bibitem{DBLP:conf/popl/BallNR03}
Ball, T., Naik, M., Rajamani, S.K.: From symptom to cause: localizing errors in
  counterexample traces. In: {POPL}. pp. 97--105. {ACM} (2003)

\bibitem{BeebeeHM-Causality-2009}
Beebee, H., Hitchcock, C., Menzie, P.: The Oxford Handbook of Causation.
  {Oxford University Press} (2009)

\bibitem{DBLP:journals/fmsd/BeerBCOT12}
Beer, I., Ben{-}David, S., Chockler, H., Orni, A., Trefler, R.J.: Explaining
  counterexamples using causality. Formal Methods in System Design
  \textbf{40}(1),  20--40 (2012)

\bibitem{DBLP:conf/ijcai/ChenPB16}
Chen, B., Pearl, J., Bareinboim, E.: Incorporating knowledge into structural
  equation models using auxiliary variables. In: {IJCAI}. pp. 3577--3583.
  {IJCAI/AAAI} Press (2016)

\bibitem{DBLP:journals/mscs/CheneyAA11}
Cheney, J., Ahmed, A., Acar, U.A.: Provenance as dependency analysis.
  Mathematical Structures in Computer Science  \textbf{21}(6),  1301--1337
  (2011)

\bibitem{DBLP:journals/jair/ChocklerH04}
Chockler, H., Halpern, J.Y.: Responsibility and blame: {A} structural-model
  approach. J. Artif. Intell. Res.  \textbf{22},  93--115 (2004)

\bibitem{DBLP:journals/tocl/ChocklerHK08}
Chockler, H., Halpern, J.Y., Kupferman, O.: What causes a system to satisfy a
  specification? {ACM} Trans. Comput. Log.  \textbf{9}(3),  20:1--20:26 (2008)

\bibitem{Westland-SEM-2015}
Christopher, W.J.: Structural Equation Models, From Paths to Networks. {Studies
  in Systems, Decision and Control 22}, {Springer} (2015)

\bibitem{CousotC77}
Cousot, P., Cousot, R.: Abstract interpretation: {A} unified lattice model for
  static analysis of programs by construction or approximation of fixpoints.
  In: {POPL}. pp. 238--252. {ACM} (1977)

\bibitem{CousotC79}
Cousot, P., Cousot, R.: Systematic design of program analysis frameworks. In:
  {POPL}. pp. 269--282. {ACM} Press (1979)

\bibitem{DBLP:journals/jlp/CousotC92}
Cousot, P., Cousot, R.: Abstract interpretation and application to logic
  programs. J. Log. Program.  \textbf{13}(2{\&}3),  103--179 (1992)

\bibitem{DBLP:conf/pldi/DilligDA12}
Dillig, I., Dillig, T., Aiken, A.: Automated error diagnosis using abductive
  inference. In: {PLDI}. pp. 181--192. {ACM} (2012)

\bibitem{DBLP:conf/uss/FranklePSGW18}
Frankle, J., Park, S., Shaar, D., Goldwasser, S., Weitzner, D.J.: Practical
  accountability of secret processes. In: {USENIX} Security Symposium. pp.
  657--674. {USENIX} Association (2018)

\bibitem{DBLP:conf/sp/GoguenM82a}
Goguen, J.A., Meseguer, J.: Security policies and security models. In: {IEEE}
  Symposium on Security and Privacy. pp. 11--20. {IEEE} Computer Society (1982)

\bibitem{DBLP:conf/sas/GreitschusDP17}
Greitschus, M., Dietsch, D., Podelski, A.: Loop invariants from
  counterexamples. In: {SAS}. Lecture Notes in Computer Science, vol. 10422,
  pp. 128--147. Springer (2017)

\bibitem{DBLP:journals/entcs/GriesmayerSB07}
Griesmayer, A., Staber, S., Bloem, R.: Automated fault localization for {C}
  programs. Electr. Notes Theor. Comput. Sci.  \textbf{174}(4),  95--111 (2007)

\bibitem{DBLP:journals/sttt/GroceCKS06}
Groce, A., Chaki, S., Kroening, D., Strichman, O.: Error explanation with
  distance metrics. {STTT}  \textbf{8}(3),  229--247 (2006)

\bibitem{Halpern-Pearl-2001}
Halpern, J.Y., Pearl, J.: Causes and explanations: {A} structural-model
  approach: Part 1: Causes. In: {UAI}. pp. 194--202. Morgan Kaufmann (2001)

\bibitem{Halpern-Pearl-2005}
Halpern, J.Y., Pearl, J.: Causes and explanations: A structural-model approach.
  part i: Causes. The British journal for the philosophy of science
  \textbf{56}(4),  843--887 (2005)

\bibitem{DBLP:conf/icse/HangalL02}
Hangal, S., Lam, M.S.: Tracking down software bugs using automatic anomaly
  detection. In: {ICSE}. pp. 291--301. {ACM} (2002)

\bibitem{DBLP:conf/cav/HeizmannHP13}
Heizmann, M., Hoenicke, J., Podelski, A.: Software model checking for people
  who love automata. In: {CAV}. Lecture Notes in Computer Science, vol.~8044,
  pp. 36--52. Springer (2013)

\bibitem{Hume-1748}
Hume, D.: An enquiry concerning human understanding. London: A. Millar (1748),
  \url{http://www.davidhume.org/texts/ehu.html}

\bibitem{DBLP:conf/esorics/JagadeesanJPR09}
Jagadeesan, R., Jeffrey, A., Pitcher, C., Riely, J.: Towards a theory of
  accountability and audit. In: {ESORICS}. Lecture Notes in Computer Science,
  vol.~5789, pp. 152--167. Springer (2009)

\bibitem{DBLP:conf/tacas/JinRS02}
Jin, H., Ravi, K., Somenzi, F.: Fate and free will in error traces. In:
  {TACAS}. Lecture Notes in Computer Science, vol.~2280, pp. 445--459. Springer
  (2002)

\bibitem{DBLP:conf/pldi/JoseM11}
Jose, M., Majumdar, R.: Cause clue clauses: error localization using maximum
  satisfiability. In: {PLDI}. pp. 437--446. {ACM} (2011)

\bibitem{DBLP:conf/sigsoft/KingJJS08}
King, D., Jaeger, T., Jha, S., Seshia, S.A.: Effective blame for
  information-flow violations. In: {SIGSOFT} {FSE}. pp. 250--260. {ACM} (2008)

\bibitem{DBLP:journals/infsof/KorelR98}
Korel, B., Rilling, J.: Dynamic program slicing methods. Information {\&}
  Software Technology  \textbf{40}(11-12),  647--659 (1998)

\bibitem{Lewis-causation-1973}
Lewis, D.: Causation. The journal of philosophy  \textbf{70}(17),  556--567
  (1973)

\bibitem{sep-causation-counterfactual}
Menzies, P.: Counterfactual theories of causation. In: Zalta, E.N. (ed.) The
  Stanford Encyclopedia of Philosophy. Metaphysics Research Lab, Stanford
  University, winter 2017 edn. (2017)

\bibitem{Pearl-Causality-2012}
Pearl, J.: Causality: Models, Reasoning and Inference. Cambridge University
  Press, 2nd edn. (2013)

\bibitem{DBLP:conf/ecoop/PistoiaFKS05}
Pistoia, M., Flynn, R.J., Koved, L., Sreedhar, V.C.: Interprocedural analysis
  for privileged code placement and tainted variable detection. In: {ECOOP}.
  Lecture Notes in Computer Science, vol.~3586, pp. 362--386. Springer (2005)

\bibitem{DBLP:journals/jlp/Plotkin04a}
Plotkin, G.D.: A structural approach to operational semantics. J. Log. Algebr.
  Program.  \textbf{60-61},  17--139 (2004)

\bibitem{DBLP:conf/sigsoft/QiRLV09}
Qi, D., Roychoudhury, A., Liang, Z., Vaswani, K.: Darwin: an approach for
  debugging evolving programs. In: {ESEC/SIGSOFT} {FSE}. pp. 33--42. {ACM}
  (2009)

\bibitem{DBLP:conf/tacas/RaviS04}
Ravi, K., Somenzi, F.: Minimal assignments for bounded model checking. In:
  {TACAS}. Lecture Notes in Computer Science, vol.~2988, pp. 31--45. Springer
  (2004)

\bibitem{DBLP:conf/kbse/RenierisR03}
Renieris, M., Reiss, S.P.: Fault localization with nearest neighbor queries.
  In: {ASE}. pp. 30--39. {IEEE} Computer Society (2003)

\bibitem{DBLP:conf/sas/Rival05}
Rival, X.: Understanding the origin of alarms in astr{\'{e}}e. In: {SAS}.
  Lecture Notes in Computer Science, vol.~3672, pp. 303--319. Springer (2005)

\bibitem{Sliedregt-Criminal-Responsibility-2012}
van Sliedregt, E.: Individual Criminal Responsibility in International Law.
  {Oxford Monographs in International Law}, {Oxford University Press} (2012)

\bibitem{DBLP:conf/vmcai/UrbanM15}
Urban, C., Min{\'{e}}, A.: Proving guarantee and recurrence temporal properties
  by abstract interpretation. In: {VMCAI}. Lecture Notes in Computer Science,
  vol.~8931, pp. 190--208. Springer (2015)

\bibitem{DBLP:conf/esop/UrbanM18}
Urban, C., M{\"{u}}ller, P.: An abstract interpretation framework for input
  data usage. In: {ESOP}. Lecture Notes in Computer Science, vol. 10801, pp.
  683--710. Springer (2018)

\bibitem{DBLP:conf/sas/UrbanU018}
Urban, C., Ueltschi, S., M{\"{u}}ller, P.: Abstract interpretation of {CTL}
  properties. In: {SAS}. Lecture Notes in Computer Science, vol. 11002, pp.
  402--422. Springer (2018)

\bibitem{Weiser81}
Weiser, M.: Program slicing. In: {ICSE}. pp. 439--449. {IEEE} Computer Society
  (1981)

\bibitem{DBLP:journals/tse/Weiser84}
Weiser, M.: Program slicing. {IEEE} Trans. Software Eng.  \textbf{10}(4),
  352--357 (1984)

\bibitem{DBLP:journals/cacm/WeitznerABFHS08}
Weitzner, D.J., Abelson, H., Berners{-}Lee, T., Feigenbaum, J., Hendler, J.A.,
  Sussman, G.J.: Information accountability. Commun. {ACM}  \textbf{51}(6),
  82--87 (2008)

\end{thebibliography}
\end{CJK*}

\iflong
\appendix
\section{Appended Proofs}

\newcommand{\figureexampleabstractiontopallprevious}{%
\begin{figure*}[ht]
	\begin{center}
	\begin{tikzpicture}[scale=0.925,>=stealth',shorten >=1pt,auto,node distance=1.5cm,thick,initial text=]\footnotesize
	\tikzstyle{every state}=[fill=gray!15,draw=black,text=black,scale=0.9]
	
	\node[initial,state
	] at (0,2.5) (1) {${}^{\ell_1}$\llap{\raisebox{2.3em}[0pt][0pt]{$\top$}}\rlap{\raisebox{0.3em}[0pt][0pt]{\quad~~\texttt{\tiny a = input\_1();}}}};
	\node[state
	] at (3,1) (2b) {\rlap{\hskip-2em\raisebox{2.3em}[0pt][0pt]{$a,b\in\interval{-1}{1}$}}${}^{\ell_2^b}$};
	\node[state
	] at (3,4) (2a) {\rlap{\hskip-2em\raisebox{2.3em}[0pt][0pt]{$a,b\in\interval{-1}{1}$}}${}^{\ell_2^a}$};
	\node[state
	] at (6,0) (3bb) {\rlap{\hskip-2em\raisebox{2.3em}[0pt][0pt]{$a,b\in\interval{-1}{1}$}}${}^{\ell_3^{bb}}$};
	\node[state
	] at (6,2) (3ba) {\rlap{\hskip-2em\raisebox{2.3em}[0pt][0pt]{$a,b\in\interval{-1}{1}$}}${}^{\ell_3^{ba}}$};
	\node[state
	] at (6,4) (3a) {\rlap{\hskip-2em\raisebox{2.3em}[0pt][0pt]{$a,b\in\interval{-1}{1}$}}${}^{\ell_3^{a}}$};
	\node[state,accepting
	] at (9,0) (4bb) {\rlap{\hskip-2em\raisebox{2.3em}[0pt][0pt]{$a,b\in\interval{-1}{1}$}}${}^{\ell_4^{bb}}$\rlap{\qquad$\abstracttraceproperty_{\!\neg b}$}};
	\node[state,accepting
	] at (9,2) (4ba) {\rlap{\hskip-2em\raisebox{2.3em}[0pt][0pt]{$a,b\in\interval{-1}{1}$}}{${}^{\ell_4^{ba}}$\rlap{\qquad$\abstracttraceproperty_{b}$}}};
	\node[state,accepting
	] at (9,4) (4a) {\rlap{\hskip-2em\raisebox{2.3em}[0pt][0pt]{$a,b\in\interval{-1}{1}$}}${}^{\ell_4^{a}}$\rlap{\qquad$\abstracttraceproperty_{b}$}};

	\path (1)  edge[->] node {{$a=0$}} (2a)
          (2a)  edge[->] node {{}} (3a)
          (3a)  edge[->] node {{}} (4a)
          (1)  edge[->] node {\llap{\raisebox{-1.8em}[0pt][0pt]{$a\neq0$}~}} (2b)
          (2a)  edge[->] node {\raisebox{0.em}{\hskip0em\texttt{\tiny b = input\_2();}}} (3a)
          (3a)  edge[->] node {\raisebox{0em}{\hskip0em\texttt{\tiny c = (a * b);}}} (4a)
          (2b)  edge[->] node {{}} (3ba)
          (3ba) edge[->] node {\raisebox{0em}{\hskip0em\texttt{\tiny c = (a * b);}}} (4ba)
          (2b)  edge[->] node {\raisebox{1em}{\hskip-3em\texttt{\tiny b = input\_2();}}} (3bb)
          (3bb) edge[->] node {\raisebox{0em}{\hskip0em\texttt{\tiny c = (a * b);}}} (4bb)
     ;
     \node at (4.6,1.9) {{$b=0$}};
     \node at (4.6,0.1) {{$b\neq 0$}};

	\node[rectangle,draw,thin,inner sep=2pt] at (6,-0.8) (m1) {\scriptsize$\abstractobservation{}{}{{}^{\ell^{bb}_3}}=\traceproperty_{\!\neg b}$};
	\node[rectangle,draw,thin,inner sep=2pt] at (9,-0.8) (m2) {\scriptsize$\abstractobservation{}{}{{}^{\ell^{bb}_4}}=\traceproperty_{\!\neg b}$};
	\node[rectangle,draw,thin,inner sep=2pt] at (6.4,1.2) (m3) {\scriptsize$\abstractobservation{}{}{{}^{\ell^{ba}_3}}=\traceproperty_{b}$};
	\node[rectangle,draw,thin,inner sep=2pt] at (9,1.2) (m4) {\scriptsize$\abstractobservation{}{}{{}^{\ell^{ba}_4}}=\traceproperty_{b}$};
	\node[rectangle,draw,thin,inner sep=2pt] at (6,3.2) (m5) {\scriptsize$\abstractobservation{}{}{{}^{\ell^{a}_3}}=\traceproperty_{b}$};
	\node[rectangle,draw,thin,inner sep=2pt] at (9,3.2) (m6) {\scriptsize$\abstractobservation{}{}{{}^{\ell^{a}_4}}=\traceproperty_{b}$};
	\node[rectangle,draw,thin,inner sep=2pt] at (3,0.2) (m7) {\scriptsize$\abstractobservation{}{}{{}^{\ell^{b}_2}}=\maximaltracepropertytop$};
	\node[rectangle,draw,thin,inner sep=2pt] at (3,3.2) (m8) {\scriptsize$\abstractobservation{}{}{{}^{\ell^{a}_2}}=\maximaltracepropertytop$};
	\node[rectangle,draw,thin,inner sep=2pt] at (-0.3,1.7) (m8) {\scriptsize$\abstractobservation{}{}{{}^{\ell_2}}=\maximaltracepropertytop$};
\end{tikzpicture}%
\end{center}
\caption{Floyd-Hoare automaton for Example \ref{ex:abstraction-top-all-previous}\label{Floyd-Hoare-automaton-for-ex:abstraction-top-all-previous}}
\end{figure*}}

\subsection{Proof of Galois Isomorphism (\ref{eq:alphapropertytransform})\label{subsec:proof-eq-alphapropertytransform}}
\begin{eqntabular}{c}
\pair{\wp(\maximalsemantics{})}{\subseteq}\GaloiS{\alphapropertytransform{\maximalsemantics{}}{}}{\gammapropertytransform{\maximalsemantics{}}{}}\pair{\alphapropertysettransform{\maximalsemantics{}}{\wp(\maximalsemantics{})}}{\subseteq}\nonumber
\end{eqntabular}
\begin{proof}
First, we prove that $\alphapropertytransform{\maximalsemantics{}}{}$ and $\gammapropertytransform{\maximalsemantics{}}{}$ are increasing.
\begin{calculus}
\hyphen{6}\formula{\traceproperty\subseteq\traceproperty'}\\
$\Rightarrow$
\formulaexplanation{(\trace'\in\traceproperty)\Rightarrow(\trace'\in\traceproperty')}{def.\ $\subseteq$}\\
$\Rightarrow$
\formulaexplanation{(\lnot(\trace\isprefix\trace')\vee(\trace'\in\traceproperty))\Rightarrow(\lnot(\trace\isprefix\trace')\vee(\trace'\in\traceproperty'))}{def.\ $\vee$}\\
$\Rightarrow$
\formulaexplanation{\{\trace\mid\forall\trace'\in\maximalsemantics{}.\; \lnot(\trace\isprefix\trace')\vee(\trace'\in\traceproperty)\}\subseteq\{\trace\mid\forall\trace'\in\maximalsemantics{}.\; \lnot(\trace\isprefix\trace')\vee(\trace'\in\traceproperty')\}}{def.\ $\subseteq$}\\
$\Rightarrow$
\formula{\{\trace\mid\forall\trace'\in\maximalsemantics{}.\; \trace\isprefix\trace' \Rightarrow \trace'\in\traceproperty\}\subseteq\{\trace\mid\forall\trace'\in\maximalsemantics{}.\; \trace\isprefix\trace' \Rightarrow \trace'\in\traceproperty'\}}\\
\rightexplanation{def.\ $\Rightarrow$}\\
$\Rightarrow$
\formula{(\prefixesofset{\traceproperty}\cap\{\trace\mid\forall\trace'\in\maximalsemantics{}.\; \trace\isprefix\trace' \Rightarrow \trace'\in\traceproperty\})\\\quad{}\subseteq(\prefixesofset{\traceproperty'}\cap\{\trace\mid\forall\trace'\in\maximalsemantics{}.\; \trace\isprefix\trace' \Rightarrow \trace'\in\traceproperty'\})}\\
\rightexplanation{def.\ $\cap$ and $\prefixesofset{}$ is increasing}\\
$\Rightarrow$
\formulaexplanation{\{\trace\in\prefixesofset{\traceproperty}\mid\forall\trace'\in\maximalsemantics{}.\; \trace\isprefix\trace' \Rightarrow \trace'\in\traceproperty\}\\\quad{}\subseteq\{\trace\in\prefixesofset{\traceproperty'}\mid\forall\trace'\in\maximalsemantics{}.\; \trace\isprefix\trace' \Rightarrow \trace'\in\traceproperty'\}}{def. $\cap$}\\
$\Rightarrow$
\formulaexplanation{\alphapropertytransform{\maximalsemantics{}}{\traceproperty}\subseteq\alphapropertytransform{\maximalsemantics{}}{\traceproperty'}}{def. $\alphapropertytransform{\maximalsemantics{}}{}$}
\end{calculus}

\begin{calculus}
\hyphen{6}\formula{\mathcal{Q}\subseteq\mathcal{Q'}}\\
$\Rightarrow$
\formulaexplanation{(\mathcal{Q}\cap\maximalsemantics{})\subseteq(\mathcal{Q'}\cap\maximalsemantics{})}{def.\ $\cap$}\\
$\Rightarrow$
\formulaexplanation{\gammapropertytransform{\maximalsemantics{}}{\mathcal{Q}}\subseteq\gammapropertytransform{\maximalsemantics{}}{\mathcal{Q'}}}{def. $\gammapropertytransform{\maximalsemantics{}}{}$}
\end{calculus}

Then, we prove that $\gammapropertytransform{\maximalsemantics{}}{}\composition\alphapropertytransform{\maximalsemantics{}}{}$ and $\alphapropertytransform{\maximalsemantics{}}{}\composition\gammapropertytransform{\maximalsemantics{}}{}$ are identity functions.
\begin{calculus}
\hyphen{6}\formula{\gammapropertytransform{\maximalsemantics{}}{}\composition\alphapropertytransform{\maximalsemantics{}}{\traceproperty}}\\
=
\formula{\gammapropertytransform{\maximalsemantics{}}{}(\{\trace\in\prefixesofset{\traceproperty}\mid\forall\trace'\in\maximalsemantics{}.\; \trace\isprefix\trace' \Rightarrow \trace'\in\traceproperty\})}\\
\rightexplanation{def.\ $\alphapropertytransform{\maximalsemantics{}}{}$}\\
=
\formula{\gammapropertytransform{\maximalsemantics{}}{}(\traceproperty\cup\{\trace\in\prefixesofset{\traceproperty}\backslash\maximalsemantics{}\mid\forall\trace'\in\maximalsemantics{}.\; \trace\isprefix\trace' \Rightarrow \trace'\in\traceproperty\})}\\\rightexplanation{$\traceproperty=\prefixesofset{\traceproperty}\cap\maximalsemantics{}$ since $\traceproperty\in\wp(\maximalsemantics{})$}\\
=
\formula{\maximalsemantics{}\cap(\traceproperty\cup\{\trace\in\prefixesofset{\traceproperty}\backslash\maximalsemantics{}\mid\forall\trace'\in\maximalsemantics{}.\; \trace\isprefix\trace' \Rightarrow \trace'\in\traceproperty\})}\\
\rightexplanation{def. $\gammapropertytransform{\maximalsemantics{}}{}$}\\
=
\formulaexplanation{\maximalsemantics{}\cap\traceproperty}{$\maximalsemantics{}\cap(\prefixesofset{\traceproperty}\backslash\maximalsemantics{})=\emptyset$}\\
=
\formulaexplanation{\traceproperty}{$\traceproperty\in\wp(\maximalsemantics{})$}
\end{calculus}

\begin{calculus}
\hyphen{6}\formula{\alphapropertytransform{\maximalsemantics{}}{}\composition\gammapropertytransform{\maximalsemantics{}}{\mathcal{Q}}}\\
=
\formula{\alphapropertytransform{\maximalsemantics{}}{}\composition\gammapropertytransform{\maximalsemantics{}}{\alphapropertytransform{\maximalsemantics{}}{\traceproperty'}}}\\
\rightexplanation{$\mathcal{Q}\in\alphapropertysettransform{\maximalsemantics{}}{\wp(\maximalsemantics{})}$, thus $\exists\traceproperty'.$ $\mathcal{Q}=\alphapropertytransform{\maximalsemantics{}}{\traceproperty'}$}\\
=
\formulaexplanation{\alphapropertytransform{\maximalsemantics{}}{\traceproperty'}}{$\gammapropertytransform{\maximalsemantics{}}{}\composition\alphapropertytransform{\maximalsemantics{}}{\traceproperty'}=\traceproperty'$}\\
=
\formulaexplanation{\mathcal{Q}}{by the assumption $\mathcal{Q}=\alphapropertytransform{\maximalsemantics{}}{\traceproperty'}$}
\end{calculus}

\hyphen{6}\quad By the above four properties, $\alphapropertytransform{\maximalsemantics{}}{}$ and $\gammapropertytransform{\maximalsemantics{}}{}$ form a Galois isomorphism.
\end{proof}

\subsection{Proofs for Lemma \ref{lemma:inquiry-decreasing}\label{subsec:proof-lemma-inquiry-decreasing}}

\begin{lemma}\label{lemma:lattice-property1} 
Given the semantics $\maximalsemantics{}$ and lattice $\latticeofmaximalproperties$ of system behaviors, for any maximal trace property $\traceproperty\in\latticeofmaximalproperties$, if a trace $\tau$ belongs to the prediction trace property that corresponds to $\traceproperty$, then every valid trace greater than $\tau$ belongs to the prediction trace property too. 
I.e. $\forall\traceproperty\in\latticeofmaximalproperties.\;\forall\tau,\tau'\in\prefixsemantics{}.\;(\tau\in\alphapropertytransform{\maximalsemantics{}}{\traceproperty}\wedge\tau\isprefix\tau')\Rightarrow\tau'\in\alphapropertytransform{\maximalsemantics{}}{\traceproperty}$.
\end{lemma}

\begin{proof} Proof by contradiction. Here we assume $\exists\traceproperty\in\latticeofmaximalproperties.$ $\exists\tau,\tau'\in\prefixsemantics{}.$ $\tau\in\alphapropertytransform{\maximalsemantics{}}{\traceproperty}\wedge\tau\isprefix\tau'\wedge\tau'\not\in\alphapropertytransform{\maximalsemantics{}}{\traceproperty}$. By the definition of prediction abstraction, $\alphapropertytransform{\maximalsemantics{}}{\traceproperty}$ $=\{\trace\in\prefixesofset{\traceproperty}$ $\mid$ $\forall\trace'\in\maximalsemantics{}.$ $\trace\isprefix\trace' \Rightarrow \trace'\in\traceproperty\}$. There are two possibilities for $\tau'\not\in\alphapropertytransform{\maximalsemantics{}}{\traceproperty}$: 
1) $\tau'\not\in\prefixesofset{\traceproperty}$, hence every maximal trace greater than $\tau'$ does not belong to $\traceproperty$;
or 2) $\exists\trace'\in\maximalsemantics{}.\; \tau'\isprefix\trace' \wedge \trace'\not\in\traceproperty$. Both cases imply that there is a maximal trace $\trace'\in\maximalsemantics{}$ such that $\tau\isprefix\tau'\isprefix\trace' \wedge \trace'\not\in\traceproperty$, which contradicts with the assumption of $\tau\in\alphapropertytransform{\maximalsemantics{}}{\traceproperty}$.
\end{proof}

\medskip
\noindent\textbf{Lemma \ref{lemma:inquiry-decreasing}.} Given the semantics $\maximalsemantics{}$ and lattice $\latticeofmaximalproperties$ of system behaviors, the corresponding inquiry function $\inquiryfunction{}{\maximalsemantics{}, \latticeofmaximalproperties}$ is decreasing on the inquired trace $\trace$: the greater $\trace$ is, the stronger property it can guarantee. I.e. $\forall\trace,\trace'\in\setofalltraces.\;\trace\isprefix\trace'\Rightarrow\inquiryfunction{}{\maximalsemantics{}, \latticeofmaximalproperties, \trace}\supseteq\inquiryfunction{}{\maximalsemantics{}, \latticeofmaximalproperties, \trace'}$.

\begin{proof}
First, if $\trace$ is invalid (i.e. $\trace\not\in\prefixsemantics{}$), then every trace $\trace'$ that is greater than $\trace$ must also be invalid (i.e. $\trace'\not\in\prefixsemantics{}$), hence $\inquiryfunction{}{\maximalsemantics{}, \latticeofmaximalproperties, \trace}=\inquiryfunction{}{\maximalsemantics{}, \latticeofmaximalproperties, \trace'}$ $=\maximaltracepropertybot$.

Second, if $\trace'\not\in\prefixsemantics{}$, then $\inquiryfunction{}{\maximalsemantics{}, \latticeofmaximalproperties, \trace'}=\maximaltracepropertybot$, hence $\inquiryfunction{}{\maximalsemantics{}, \latticeofmaximalproperties, \trace}\supseteq\maximaltracepropertybot=\inquiryfunction{}{\maximalsemantics{}, \latticeofmaximalproperties, \trace'}$.

Last, if $\trace,\trace'\in\prefixsemantics{}$, then

\begin{calculus}
\formula{\trace\isprefix\trace'}\\
$\Rightarrow$
\formulaexplanation{\forall\traceproperty\in\latticeofmaximalproperties\!.\;\trace\in\alphapropertytransform{\maximalsemantics{}}{\traceproperty}\Rightarrow\trace'\in\alphapropertytransform{\maximalsemantics{}}{\traceproperty}}{lemma. \ref{lemma:lattice-property1}}\\
$\Rightarrow$
\formula{\{\traceproperty\in\latticeofmaximalproperties\mid\trace\in\alphapropertytransform{\maximalsemantics{}}{\traceproperty}\}\subseteq\{\traceproperty\in\latticeofmaximalproperties\mid\trace'\in\alphapropertytransform{\maximalsemantics{}}{\traceproperty}\}}\\
\rightexplanation{def. $\Rightarrow$}\\
$\Rightarrow$
\formula{\maximaltracepropertymeet\{\traceproperty\in\latticeofmaximalproperties\mid\trace\in\alphapropertytransform{\maximalsemantics{}}{\traceproperty}\}\supseteq \maximaltracepropertymeet\{\traceproperty\in\latticeofmaximalproperties\mid\trace'\in\alphapropertytransform{\maximalsemantics{}}{\traceproperty}\}}\\
\rightexplanation{def. $\maximaltracepropertymeet$}\\
$\Rightarrow$
\lastformulaexplanation{\inquiryfunction{}{\maximalsemantics{}, \latticeofmaximalproperties, \trace}\supseteq\inquiryfunction{}{\maximalsemantics{}, \latticeofmaximalproperties, \trace'}}{def. $\inquiryfunction{}{}$}{\quad}
\end{calculus}
\vskip-1.8em
\end{proof}

\subsection{Proof of Corollary \ref{corollary:observation-guarantee} \label{subsec:proof-corollary-observation-guarantee}}

\noindent\textbf{Corollary \ref{corollary:observation-guarantee}.} Given the semantics $\maximalsemantics{}$ and lattice $\latticeofmaximalproperties$ of system behaviors, for any observer with cognizance $\cognizancefunction{}{}$, if the observation function maps a trace $\trace$ to a maximal trace property $\traceproperty\in\latticeofmaximalproperties$, then $\trace$ guarantees the satisfaction of property $\traceproperty$ (i.e. every valid maximal trace that is greater than or equal to $\trace$ is guaranteed to have property $\traceproperty$).

\begin{proof}
Suppose $\inquiryfunction{}{\maximalsemantics{}, \latticeofmaximalproperties, \trace}$ $=\traceproperty'$. By the corollary \ref{corollary:inquiry-guarantee}, $\trace$ guarantees the satisfaction of property $\traceproperty'$, i.e. every valid maximal trace that is greater than or equal to $\trace$ belongs to $\traceproperty'$.

In addition, by the definition of observation function in (\ref{eq:observation}), we know $\traceproperty$ $=$ $\observation{}{}{\maximalsemantics{}, \latticeofmaximalproperties, \cognizancefunction{}{}, \trace}$ $=\maximaltracepropertyjoin\{\inquiryfunction{}{\maximalsemantics{}, \latticeofmaximalproperties, \trace'}\mid\trace'\in\cognizancefunction{}{\trace}\}$ $\supseteq$ $\inquiryfunction{}{\maximalsemantics{}, \latticeofmaximalproperties, \trace}$ $=$ $\traceproperty'$, since the cognizance is extensive (i.e. $\trace\in\cognizancefunction{}{\trace}$). Therefore, every valid maximal trace that is greater than or equal to $\trace$ belongs to $\traceproperty$. That is to say, $\trace$ guarantees the satisfaction of property $\traceproperty$.
\end{proof}

\subsection{Proof of Lemma \ref{lemma:observation-decreasing} \label{subsec:proof-lemma-observation-decreasing}}
\textbf{Lemma \ref{lemma:observation-decreasing}.} Given the semantics $\maximalsemantics{}$, lattice $\latticeofmaximalproperties$ of system behaviors and cognizance function $\cognizancefunction{}{}$, the observation function $\observation{}{}{\maximalsemantics{}, \latticeofmaximalproperties, \cognizancefunction{}{}}$ is decreasing: the longer the observed trace $\trace$ is, the stronger property it can observe. 
I.e. $\forall\trace,\trace'\in\setofalltraces.\;\trace\isprefix\trace'\Rightarrow\observation{}{}{\maximalsemantics{}, \latticeofmaximalproperties, \cognizancefunction{}{}, \trace}\supseteq\observation{}{}{\maximalsemantics{}, \latticeofmaximalproperties, \cognizancefunction{}{}, \trace'}$.

\begin{proof}
First, if $\trace$ is invalid (i.e. $\trace\not\in\prefixsemantics{}$), then every trace $\trace'$ that is greater than $\trace$ must also be invalid (i.e. $\trace'\not\in\prefixsemantics{}$), hence it is easy to find that $\observation{}{}{\maximalsemantics{}, \latticeofmaximalproperties, \cognizancefunction{}{}, \trace}=\observation{}{}{\maximalsemantics{}, \latticeofmaximalproperties, \cognizancefunction{}{}, \trace}=\maximaltracepropertybot$.

Second, if $\trace'\not\in\prefixsemantics{}$, then we have $\observation{}{}{\maximalsemantics{}, \latticeofmaximalproperties, \cognizancefunction{}{}, \trace'}=\maximaltracepropertybot$. Hence, it is trivial to find $\observation{}{}{\maximalsemantics{}, \latticeofmaximalproperties, \cognizancefunction{}{}, \trace}\supseteq\maximaltracepropertybot=\observation{}{}{\maximalsemantics{}, \latticeofmaximalproperties, \cognizancefunction{}{}, \trace'}$.

Last, if $\trace,\trace'\in\prefixsemantics{}$, then $\trace$ must be a valid non-maximal trace, i.e. $\trace\in\prefixsemantics{}\backslash\maximalsemantics{}$. From Lemma \ref{lemma:monitor-one-trace}, it is easy to see $\forall\event\in\setofallevents.$ $\observation{}{}{\maximalsemantics{}, \latticeofmaximalproperties, \cognizancefunction{}{}, \trace}\supseteq\observation{}{}{\maximalsemantics{}, \latticeofmaximalproperties, \cognizancefunction{}{}, \trace\event}$. Since $\trace'$ is a prolongation of $\trace$ with events, by the transitivity of $\supseteq$, we can prove that $\observation{}{}{\maximalsemantics{}, \latticeofmaximalproperties, \cognizancefunction{}{}, \trace}\supseteq\observation{}{}{\maximalsemantics{}, \latticeofmaximalproperties, \cognizancefunction{}{}, \trace'}$.
\end{proof}

\subsection{Proofs for Lemma \ref{lemma:responsibility-definition}\label{subsec:proof-lemma-responsibility-definition}}

\begin{lemma}\label{lemma:lattice-property2} Given the semantics $\maximalsemantics{}$and the lattice $\latticeofmaximalproperties$ of system behaviors, for any maximal trace property $\traceproperty\in\latticeofmaximalproperties$ and any valid prefix trace $\tau$ that is not maximal, if every valid prefix trace $\tau\event$ which concatenates $\tau$ with a new event $\event$ belongs to the prediction trace property $\alphapropertytransform{\maximalsemantics{}}{\traceproperty}$, then $\tau$ belongs to $\alphapropertytransform{\maximalsemantics{}}{\traceproperty}$ too. I.e. $\forall\traceproperty\in\latticeofmaximalproperties.$ $\forall\tau\in\prefixsemantics{}\backslash\maximalsemantics{}.$ $(\forall\tau\event\in\prefixsemantics{}.$ $\tau\event\in\alphapropertytransform{\maximalsemantics{}}{\traceproperty})\Rightarrow\tau\in\alphapropertytransform{\maximalsemantics{}}{\traceproperty}$.
\end{lemma}

\begin{proof} Prove by contradiction. Assume that $\exists\traceproperty\in\latticeofmaximalproperties.$ $\exists\tau\in\prefixsemantics{}\backslash\maximalsemantics{}.$ $(\forall\tau\event\in\prefixsemantics{}.\;\tau\event\in\alphapropertytransform{\maximalsemantics{}}{\traceproperty})\wedge\tau\not\in\alphapropertytransform{\maximalsemantics{}}{\traceproperty}$. By the definition that $\alphapropertytransform{\maximalsemantics{}}{\traceproperty}=\{\trace\in\prefixesofset{\traceproperty}\mid\forall\trace'\in\maximalsemantics{}.$ $\trace\isprefix\trace' \Rightarrow \trace'\in\traceproperty\}$. In order to have $\tau\not\in\alphapropertytransform{\maximalsemantics{}}{\traceproperty}$, there are two possibilities.

1) $\tau\not\in\prefixesofset{\traceproperty}$. This implies that $\forall\tau\event\in\prefixsemantics{}.$ $\tau\event\not\in\prefixesofset{\traceproperty}$, which further implies that $\forall\tau\event\in\prefixsemantics{}.$ $\tau\event\not\in\alphapropertytransform{\maximalsemantics{}}{\traceproperty}$. Since $\tau\in\prefixsemantics{}\backslash\maximalsemantics{}$, there must exist at least one $\event$ such that $\tau\event\in\prefixsemantics{}\wedge\tau\event\not\in\alphapropertytransform{\maximalsemantics{}}{\traceproperty}$.

2) There is a maximal trace $\trace'\in\maximalsemantics{}$ such that $\tau\prec\trace' \wedge \trace'\not\in\traceproperty$. Take $\event=\trace'_{|\tau|}$, then $\tau\event\in\prefixsemantics{}\wedge$ $\tau\event\isprefix\trace' \wedge \trace'\not\in\traceproperty$ holds, which implies $\tau\event\in\prefixsemantics{}\wedge$ $\tau\event\not\in\alphapropertytransform{\maximalsemantics{}}{\traceproperty}$. 

Both two cases find that $\exists\tau\event\in\prefixsemantics{}.$ $\tau\event\not\in\alphapropertytransform{\maximalsemantics{}}{\traceproperty}$, which contradicts with the assumption $\forall\tau\event\in\prefixsemantics{}.$ $\tau\event\in\alphapropertytransform{\maximalsemantics{}}{\traceproperty}$.
\end{proof}

\begin{corollary}\label{corollary:cognizance-property} For any cognizance $\cognizancefunction{}{}$, we have $\underset{\event\in\setofallevents}{\bigmaximaltracepropertyjoin}\cognizancefunction{}{\event}\supseteq\setofallevents$.
\end{corollary}
\begin{proof}This corollary comes immediately from the fact that the cognizance function $\cognizancefunction{}{}$ is extensive.
\end{proof}

\begin{corollary}\label{corollary:observation-property} Given the semantics $\maximalsemantics{}$ and lattice $\latticeofmaximalproperties$ of behaviors, $\forall\trace\in\prefixsemantics{}\backslash\maximalsemantics{}.$ $\inquiryfunction{}{\maximalsemantics{}, \latticeofmaximalproperties, \trace}=\underset{\event\in\setofallevents}{\bigmaximaltracepropertyjoin}\inquiryfunction{}{\maximalsemantics{}, \latticeofmaximalproperties, \trace\event}=\underset{\trace\event\in\prefixsemantics{}}{\bigmaximaltracepropertyjoin}\inquiryfunction{}{\maximalsemantics{}, \latticeofmaximalproperties, \trace\event}$.
\end{corollary}
\begin{proof}
First, $\maximaltracepropertyjoin\{\inquiryfunction{}{\maximalsemantics{}, \latticeofmaximalproperties, \trace\event}\mid\event\in\setofallevents\}$ $=(\maximaltracepropertyjoin\{\inquiryfunction{}{\maximalsemantics{}, \latticeofmaximalproperties, \trace\event}\mid\trace\event\in\prefixsemantics{}\})$ $\maximaltracepropertyjoin(\maximaltracepropertyjoin\{\inquiryfunction{}{\maximalsemantics{}, \latticeofmaximalproperties, \trace\event}\mid\trace\event\not\in\prefixsemantics{}\})$ $=(\maximaltracepropertyjoin\{\inquiryfunction{}{\maximalsemantics{}, \latticeofmaximalproperties, \trace\event}\mid\trace\event\in\prefixsemantics{}\})\maximaltracepropertyjoin\maximaltracepropertybot$ $=\maximaltracepropertyjoin\{\inquiryfunction{}{\maximalsemantics{}, \latticeofmaximalproperties, \trace\event}\mid\trace\event\in\prefixsemantics{}\}$.

Second, we prove $\maximaltracepropertyjoin\{\inquiryfunction{}{\maximalsemantics{}, \latticeofmaximalproperties, \trace\event}\mid\trace\event\in\prefixsemantics{}\}$ $=$ $\inquiryfunction{}{\maximalsemantics{}, \latticeofmaximalproperties, \trace}$ in two steps: 
1) by lemma. \ref{lemma:inquiry-decreasing}, it is proved that $\forall\trace, \trace\event\in\setofalltraces.\;\inquiryfunction{}{\maximalsemantics{}, \latticeofmaximalproperties, \trace}\supseteq\inquiryfunction{}{\maximalsemantics{}, \latticeofmaximalproperties, \trace\event}$, thus $\inquiryfunction{}{\maximalsemantics{}, \latticeofmaximalproperties, \trace}\supseteq\maximaltracepropertyjoin\{\inquiryfunction{}{\maximalsemantics{}, \latticeofmaximalproperties, \trace\event}\mid\trace\event\in\prefixsemantics{}\}$. 
2) assume $\inquiryfunction{}{\maximalsemantics{}, \latticeofmaximalproperties, \trace}\supsetneq$ $\maximaltracepropertyjoin\{\inquiryfunction{}{\maximalsemantics{}, \latticeofmaximalproperties, \trace\event}\mid\trace\event\in\prefixsemantics{}\}$ $=$ $\traceproperty$. By the definition of $\inquiryfunction{}{}$ in (\ref{eq:inquiryfunction}), we know that $\trace\not\in\alphapropertytransform{\maximalsemantics{}}{\traceproperty}$ and $\forall\trace\event\in\prefixsemantics{}.$ $\trace\event\in\alphapropertytransform{\maximalsemantics{}}{\traceproperty}$, which is impossible by lemma. \ref{lemma:lattice-property2}. Thus, by contradiction, $\inquiryfunction{}{\maximalsemantics{}, \latticeofmaximalproperties, \trace}$ $=$ $\maximaltracepropertyjoin\{\inquiryfunction{}{\maximalsemantics{}, \latticeofmaximalproperties, \trace\event}\mid\trace\event\in\prefixsemantics{}\}$.
\end{proof}

\begin{lemma}\label{lemma:monitor-one-trace} Given the semantics $\maximalsemantics{}$, the lattice $\latticeofmaximalproperties$ of system behaviors and the cognizance function $\cognizancefunction{}{}$, we have: 

\noindent$\forall\trace\in\prefixsemantics{}\backslash\maximalsemantics{}.$ $\observation{}{}{\maximalsemantics{}, \latticeofmaximalproperties, \cognizancefunction{}{}, \trace}$ \\
$=\underset{\event\in\setofallevents}{\bigmaximaltracepropertyjoin}\observation{}{}{\maximalsemantics{}, \latticeofmaximalproperties, \cognizancefunction{}{}, \trace\event}=\underset{\trace\event\in\prefixsemantics{}}{\bigmaximaltracepropertyjoin}\observation{}{}{\maximalsemantics{}, \latticeofmaximalproperties, \cognizancefunction{}{}, \trace\event}$.
\end{lemma}

\begin{proof}We start from the right side.
\begin{calculus}
\formula{\underset{\trace\event\in\prefixsemantics{}}{\maximaltracepropertyjoin}\observation{}{}{\maximalsemantics{}, \latticeofmaximalproperties, \cognizancefunction{}{}, \trace\event}}\\
=
\formulaexplanation{(\underset{\trace\event\in\prefixsemantics{}}{\maximaltracepropertyjoin}\observation{}{}{\maximalsemantics{}, \latticeofmaximalproperties, \cognizancefunction{}{}, \trace\event})\maximaltracepropertyjoin\maximaltracepropertybot}{def.$\maximaltracepropertybot$}\\
=
\formula{\underset{\trace\event\in\prefixsemantics{}}{\maximaltracepropertyjoin}\observation{}{}{\maximalsemantics{}\!, \latticeofmaximalproperties\!, \cognizancefunction{}{}, \trace\event}\maximaltracepropertyjoin\underset{\trace\event\not\in\prefixsemantics{}}{\maximaltracepropertyjoin}\observation{}{}{\maximalsemantics{}\!, \latticeofmaximalproperties\!, \cognizancefunction{}{}, \trace\event}}\\
=
\formulaexplanation{\underset{\event\in\setofallevents}{\maximaltracepropertyjoin}\observation{}{}{\maximalsemantics{}, \latticeofmaximalproperties, \cognizancefunction{}{}, \trace\event}}{merge two cases}\\
=
\formulaexplanation{\maximaltracepropertyjoin\{\inquiryfunction{}{\maximalsemantics{}, \latticeofmaximalproperties, \tau}\mid\tau\in\cognizancefunction{}{\trace\event}\wedge\event\in\setofallevents\}}{def.\ $\observationforsingletrace{}{}{}$}\\
=
\formulaexplanation{\maximaltracepropertyjoin\{\inquiryfunction{}{\maximalsemantics{}, \latticeofmaximalproperties, \trace'\trace''}\mid\trace'\trace''\in\cognizancefunction{}{\trace\event}\wedge\event\in\setofallevents\}}{replace $\tau$ with $\trace'\trace''$}\\
=
\formulaexplanation{\maximaltracepropertyjoin\{\inquiryfunction{}{\maximalsemantics{}, \latticeofmaximalproperties, \trace'\trace''}\mid\trace'\in\cognizancefunction{}{\trace}\wedge\trace''\in\cognizancefunction{}{\event}\wedge\event\in\setofallevents\}}{assumption \ref{assumption:cognizance1}}\\
=
\formula{(\maximaltracepropertyjoin\{\inquiryfunction{}{\maximalsemantics{}, \latticeofmaximalproperties, \trace'\trace''}\mid\trace'\in\cognizancefunction{}{\trace}\wedge\trace''\in\cognizancefunction{}{\event}\wedge\event\in\setofallevents\wedge|\trace''|=1\})\;\maximaltracepropertyjoin\;(\maximaltracepropertyjoin\{\inquiryfunction{}{\maximalsemantics{}, \latticeofmaximalproperties, \trace'\trace''}\mid\trace'\in\cognizancefunction{}{\trace}\wedge\trace''\in\cognizancefunction{}{\event}\wedge\event\in\setofallevents\wedge|\trace''|=0\})\;\maximaltracepropertyjoin\;(\maximaltracepropertyjoin\{\inquiryfunction{}{\maximalsemantics{}, \latticeofmaximalproperties, \trace'\trace''}\mid\trace'\in\cognizancefunction{}{\trace}\wedge\trace''\in\cognizancefunction{}{\event}\wedge\event\in\setofallevents\wedge|\trace''|>1\})}\\
\rightexplanation{split into 3 cases by the length of $\trace''$}
\end{calculus}

The first case:
\begin{calculus}
\formula{\maximaltracepropertyjoin\{\inquiryfunction{}{\maximalsemantics{}, \latticeofmaximalproperties, \trace'\trace''}\mid\trace'\in\cognizancefunction{}{\trace}\wedge\trace''\in\cognizancefunction{}{\event}\wedge\event\in\setofallevents\wedge|\trace''|=1\}}\\
=
\formulaexplanation{\maximaltracepropertyjoin\{\inquiryfunction{}{\maximalsemantics{}, \latticeofmaximalproperties, \trace'\trace''}\mid\trace'\in\cognizancefunction{}{\trace}\wedge\trace''\in\setofallevents\}}{corol. \ref{corollary:cognizance-property}}\\
=
\formulaexplanation{\maximaltracepropertyjoin\{\inquiryfunction{}{\maximalsemantics{}, \latticeofmaximalproperties, \trace'}\mid\trace'\in\cognizancefunction{}{\trace}\}}{corollary \ref{corollary:observation-property}}\\
=
\formulaexplanation{\observation{}{}{\maximalsemantics{}, \latticeofmaximalproperties, \cognizancefunction{}{}, \trace}}{def.\ $\observationforsingletrace{}{}{}$}
\end{calculus}

The second case: if there exists $\event\in\setofallevents$ s.t. $\emptytrace\in\cognizancefunction{}{\event}$, then $\maximaltracepropertyjoin\{\inquiryfunction{}{\maximalsemantics{}, \latticeofmaximalproperties, \trace'\trace''}$ $\mid\trace'\in\cognizancefunction{}{\trace}\wedge\trace''\in\cognizancefunction{}{\event}\wedge\event\in\setofallevents\wedge|\trace''|=0\} 
=
\maximaltracepropertyjoin\{\inquiryfunction{}{\maximalsemantics{}, \latticeofmaximalproperties, \trace'}\mid\trace'\in\cognizancefunction{}{\trace}\}
=
\observation{}{}{\maximalsemantics{}, \latticeofmaximalproperties, \cognizancefunction{}{}, \trace}$. Otherwise, it is an empty set.

The third case: 

\begin{calculus}
\formula{\maximaltracepropertyjoin\{\inquiryfunction{}{\maximalsemantics{}, \latticeofmaximalproperties, \trace'\trace''}\mid\trace'\in\cognizancefunction{}{\trace}\wedge\trace''\in\cognizancefunction{}{\event}\wedge\event\in\setofallevents\wedge|\trace''|>1\}}\\
$\subseteq$
\formula{\maximaltracepropertyjoin\{\inquiryfunction{}{\maximalsemantics{}, \latticeofmaximalproperties, \trace'}\mid\trace'\in\cognizancefunction{}{\trace}\wedge\trace''\in\cognizancefunction{}{\event}\wedge\event\in\setofallevents\wedge|\trace''|>1\}}\\
\rightexplanation{$\inquiryfunction{}{\maximalsemantics{}, \latticeofmaximalproperties, \trace'\trace''}\subseteq\inquiryfunction{}{\maximalsemantics{}, \latticeofmaximalproperties, \trace'}$}\\
$\subseteq$
\formulaexplanation{\maximaltracepropertyjoin\{\inquiryfunction{}{\maximalsemantics{}, \latticeofmaximalproperties, \trace'}\mid\trace'\in\cognizancefunction{}{\trace}\}}{def. $\maximaltracepropertyjoin$}\\
=
\formulaexplanation{\observation{}{}{\maximalsemantics{}, \latticeofmaximalproperties, \cognizancefunction{}{}, \trace}}{def.\ $\observationforsingletrace{}{}{}$}
\end{calculus}

Joining the above three cases together, we have proved that \\
\noindent$\underset{\trace\event\in\prefixsemantics{}}{\bigmaximaltracepropertyjoin}\observation{}{}{\maximalsemantics{}, \latticeofmaximalproperties, \cognizancefunction{}{}, \trace\event}=\observation{}{}{\maximalsemantics{}, \latticeofmaximalproperties, \cognizancefunction{}{}, \trace}$.
\end{proof}

\bigskip
\noindent\textbf{Lemma \ref{lemma:responsibility-definition}.}  If $\traceresponsible$ is said to be responsible for a behavior $\behavior$ in a valid trace $\tracehistory\traceresponsible\tracefuture$, then $\tracehistory\traceresponsible$ guarantees the occurrence of behavior $\behavior$, and there must exist another valid prefix trace $\tracehistory\traceresponsible'$ such that the behavior $\behavior$ is not guaranteed.

\begin{proof}
First, from the definition of responsiblity, we know $\observation{}{}{\maximalsemantics{}, \latticeofmaximalproperties, \cognizancefunction{}{}, \tracehistory\traceresponsible}$ $\subseteq\behavior$. By corollary \ref{corollary:observation-guarantee}, $\tracehistory\traceresponsible$ guarantees the satisfaction of $\observation{}{}{\maximalsemantics{}, \latticeofmaximalproperties, \cognizancefunction{}{}, \tracehistory\traceresponsible}$, which is at least as strong as $\behavior$. Thus, the occurrence of behavior $\behavior$ is guaranteed.

Second, we prove by contradiction. Assume that every valid trace $\tracehistory\traceresponsible'$ guarantees the occurrence of behavior $\behavior$ (i.e. $\forall\tracehistory\traceresponsible'\in\prefixsemantics{}.$ $\observation{}{}{\maximalsemantics{}, \latticeofmaximalproperties, \cognizancefunction{}{}, \tracehistory\traceresponsible'}$ $\subseteq\behavior$). By lemma \ref{lemma:monitor-one-trace}, we can prove that $\observation{}{}{\maximalsemantics{}, \latticeofmaximalproperties, \cognizancefunction{}{}, \tracehistory}$ $\subseteq\behavior$, which contradicts with the requirement term of $\behavior\subsetneq\observation{}{}{\maximalsemantics{}, \latticeofmaximalproperties, \cognizancefunction{}{}, \tracehistory}$ for $\traceresponsible$ to be responsible for the behavior $\behavior$.
\end{proof}

\section{Abstract Responsibility Analysis\label{sec:abstract-responsibility}}
In general the concrete maximal trace semantics $\maximalsemantics{}$ is not computable, thus the static responsibility analysis $\alpharesponsibility{}{}{\maximalsemantics{}, \latticeofmaximalproperties, \cognizancefunction{}{}, \behavior, \tracepropertyT}$ proposed in Section \ref{sec:formal-definition} is undecidable, and an implementation of it has to abstract sets of finite or infinite traces involved in $\maximalsemantics{}$, $\latticeofmaximalproperties$, $\cognizancefunction{}{}$, $\behavior$, and $\tracepropertyT$. 



\subsection{Program Operational Semantics}\label{sect:Program-operational-semantics}

For any given program, we assume that the program semantics is fixed in the static analyzer and defined by a small-step operational semantics \cite{DBLP:journals/jlp/Plotkin04a} $\quadruple{\Sigma}{I}{A}{\tau}$
specifying the program states $\Sigma$, the initial states $I\subseteq\Sigma$, the actions $A$, and the transition relation $\tau\in\wp(\Sigma\times A\times \Sigma)$. 
This is used to generate a trace operational semantics $\tracesemantics$, which is the set of maximal finite or infinite traces starting from an
initial state in $I$ and such that any two consecutive states are separated by an action as specified by $\tau$.

The events are chosen to be transitions $s\transition{a}s'\in\tau$. An isomorphic abstraction of the trace operational semantics $\tracesemantics$
where states are erased and events $\triple{s}{a}{s'}$ 
are preserved provides the concrete maximal trace semantics $\maximalsemantics{}$ 
where traces are sequences of events as assumed in Section \ref{subsec:event-model}. 

\subsection{Static Invariance and Eventuality Analyzers}\label{sect:static-analyzer}
Our responsibility analysis is designed so as to reuse a static invariance analyzer, which is able to attach to each program
point an abstract invariant in an abstract domain $\quadruple{\abstractdomain}{\sqsubseteq}{\bot}{\sqcup}$ 
where the concretization of $\bot$ is \textit{false}/$\emptyset$.
Such a static analyzer over-approximates states which are reachable from an initial abstract 
invariant (e.g.\ the initial states) and may reach a final abstract invariant (e.g.\
the final states), by iterating fixpoint forward/backward static analyses in the 
abstract domain $\pair{\abstractdomain}{\sqsubseteq}$, if necessary with widening 
extrapolation/narrowing interpolation \cite[Section 6]{DBLP:journals/jlp/CousotC92}. 
To be more precise, the forward analysis uses a least fixpoint. An option of the static analyzer is offered to specify whether programs are assumed to terminate (thus excluding infinite traces) or not. Therefore the backward 
analysis uses a least fixpoint if termination is optionally required and a greatest 
fixpoint to allow for non-termination.

If available, an existing eventuality static analyzer such as \cite{DBLP:conf/vmcai/UrbanM15,DBLP:conf/sas/UrbanU018}
can be used to increase the precision of the responsibility analysis. Such an eventuality static analyzer should
be able to prove that there exists a concrete execution which, starting from a reachable state satisfying an abstract invariant, will
definitely reach a terminal state satisfying an abstract property (or that this is impossible).

\subsection{User Specification for the Responsibility Analysis\label{sect:User-specification}}
In general, for a static responsibility analysis $\alpharesponsibility{}{}{}(\maximalsemantics{},$ $\latticeofmaximalproperties,$ $\cognizancefunction{}{}, \behavior, \tracepropertyT)$
where $\maximalsemantics{}$ is defined as in Section \ref{sect:Program-operational-semantics}, 
the user must specify the lattice of properties of interest $\latticeofmaximalproperties$, the cognizance $\cognizancefunction{}{}$, the behavior
$\behavior$ and traces $\tracepropertyT$ of interest. 

One possibility is to design a static analyzer where $\latticeofmaximalproperties, \cognizancefunction{}{}$ and $\tracepropertyT$ are fixed, while $\behavior$ is given by options. Another possibility is to design a parameterized analyzer with parameters specified by the user.

Specifically, here we consider the most frequent case of $\latticeofmaximalproperties$, which includes a trace property $\traceproperty_{b}$ 
and its complement $\traceproperty_{\!\neg b}$. The behavior of interest $\behavior$ is assumed to be $\traceproperty_{b}$ (or alternatively $\traceproperty_{\!\neg b}$). Thus, the user needs to specify the sets of traces 
$\traceproperty_{b}$, $\traceproperty_{\!\neg b}$, and $\tracepropertyT$, as well as the cognizance $\cognizancefunction{}{}$.

\vskip-.5em
\begin{eqntabular*}{C}
\begin{tikzpicture}[scale=0.7]
\footnotesize
     \draw (0,1) node[name=title] {{\llap{$\latticeofmaximalproperties$\ :\qquad}}};
     \draw (1,0) node[name=b] {{$\maximaltracepropertybot\rlap{ = $\emptyset$}$}};
     \draw (0,1) node[name=lb] {{$\traceproperty_{b}$}};
     \draw (2,1) node[name=lnb] {{$\traceproperty_{\!\neg b}$}};
     \draw (1,2) node[name=t] {{$\maximaltracepropertytop\rlap{ = $ \maximalsemantics{}$}$}};
     \draw (0,-1) node[name=fictive] {};
     \draw[semithick] (b) -- (lb);
     \draw[semithick] (b) -- (lnb);
     \draw[semithick] (lb) -- (t);
     \draw[semithick] (lnb) -- (t);
\end{tikzpicture}
\end{eqntabular*}
\vskip-2.5em


The user specification needs to be given in the abstract domain $\pair{\abstractdomain}{\sqsubseteq}$. More precisely, $\traceproperty_{b}$, $\traceproperty_{\!\neg b}$, and $\tracepropertyT$ are specified by abstract invariants in $\abstractdomain$ attached to each program point, which represent
sets of maximal concrete traces such that the abstract invariant attached to a program point holds for the program state whenever
the trace reaches that program point. If termination is requested only the finite traces are considered.

The cognizance function $\cognizancefunction{}{}$ is specified by abstract relational invariants in $\pair{\abstractdomain}{\sqsubseteq}$ expressing relational properties about two executions of a single program on different inputs \cite{DBLP:journals/pacmpl/AguirreBG0S17}. Such a relational invariant abstracts the relation $\{\pair{\trace}{\trace'}\mid\trace'\in\cognizancefunction{}{\trace}\}$ between traces, in which the states of ${\trace}$ and ${\trace'}$ at any program point ${}^{\ell}$ satisfy the invariant relation attached to that point.

\begin{example}\label{ex:abstraction}Let us take the following program as an example of the abstract responsibility analysis for the bug $\mathtt{c}==0$ at line ${}^{\ell_4}$ with the omniscient cognizance.\par
\bgroup\ttfamily\footnotesize\begin{eqntabular*}[fl]{@{\qquad}L}
${}^{\ell_1}$  a = input\_1(); // Input 1 or -1\\
${}^{\ell_2}$  b = input\_2(); // Input 1 or -1\\
${}^{\ell_3}$  c = (a - b);\\
${}^{\ell_4}$  // c==0 is a bug
\end{eqntabular*}\egroup

In this example, we assume that the abstract domain $\pair{\abstractdomain}{\sqsubseteq}$ chosen for the static invariance analyzer and user specifications can express
the reduced product of intervals and symbolic (dis)equalities. 
The abstract user specifications are as follows (the superscript $\userspec{}$ stands for user specification):\par
\bgroup\ttfamily\abovedisplayskip6pt\belowdisplayskip2pt\footnotesize\abovedisplayskip0.75\abovedisplayskip\belowdisplayskip0.5\belowdisplayskip\arraycolsep0.75\arraycolsep%
\begin{eqntabular*}{c|c|c|c|c|L@{}}
&{}^{\ell_1}&{}^{\ell_2}&{}^{\ell_3}&{}^{\ell_4}\\\hline
\ulstrut\userspec{\tracepropertyT}&\mathsf{true}&\mathsf{true}&\mathsf{true}&\mathsf{true}&\normalfont\footnotesize traces of interest\\
 \userspec{\traceproperty_{b}}&\mathsf{true}&\mathsf{true}&\mathsf{true}&\mathtt{c}\mathrel{{=}{=}}0&\normalfont\footnotesize erroneous traces\\
\userspec{\traceproperty_{\!\neg b}}&\mathsf{true}&\mathsf{true}&\mathsf{true}&\mathtt{c}\mathrel{{!}{=}}0&\normalfont\footnotesize correct traces\\
\lLstrut\userspec{\omniscientcognizance{}{}}&\mathsf{true}&\begin{array}[t]{@{}c@{}}\mathtt{a}\mathrel{{=}{=}}\mathtt{a}'\\{}\in\interval{-1}{1}\end{array}&\begin{array}[t]{@{}c@{}}\mathtt{b}\mathrel{{=}{=}}\mathtt{b}'\\{}\in\interval{-1}{1}\end{array}&\mathsf{true}&\normalfont\footnotesize omniscient cognizance\\
\lLstrut\userspec{\nonomniscientcognizance{}{}}&\mathsf{true}&\mathsf{true}&\begin{array}[t]{@{}c@{}}\mathtt{b}\mathrel{{=}{=}}\mathtt{b}'\\{}\in\interval{-1}{1}\end{array}&\mathsf{true}&\normalfont\footnotesize\begin{tabular}[t]{@{}c@{}}non-omniscient cogni-\\zance of 2${}^{\mathrm{nd}}$ input\end{tabular}\\\hline
\end{eqntabular*}
\egroup
In the above table, $\ulstrut\userspec{\tracepropertyT}$ specifies that all valid traces are analyzed, $\userspec{\traceproperty_{b}}$ specifies the set of erroneous traces in which $\mathtt{c}$ is 0 at ${}^{\ell_4}$, and $\userspec{\traceproperty_{\!\neg b}}$ specifies the set of correct traces.
In addition, two example cases for the cognizance are given. For the omniscient observer with cognizance $\omniscientcognizance{}{}$, the input at ${}^{\ell_2}$ is done knowing the value of $\mathtt{a}$, thus only identical traces are related by the invariant; for a non-omniscient observer with cognizance $\nonomniscientcognizance{}{}$, the input at ${}^{\ell_2}$ is independent of the input at ${}^{\ell_1}$, thus traces with identical values of $\mathtt{b}$ are related in the invariant, regardless of the value of $\mathtt{a}$.
\qed\end{example}

\subsection{Strengthening the User Specification by the Abstract Invariance Semantics}

The objective of the invariance static analysis is to abstract the program concrete semantics $\semantics{}$ of Section \ref{sect:Program-operational-semantics} and to strengthen the user specifications of Section \ref{sect:User-specification}. 
Let $\abstractinvariancesemantics{}$ be the abstract invariance semantics over-approximating $\semantics{}$ computed by the static analyzer of Section \ref{sect:static-analyzer}.
Since the user specifications $\userspec{\tracepropertyT}$,  $\userspec{\traceproperty_{b}}$,
and $\userspec{\traceproperty_{\!\neg b}}$ are relative to the concrete trace semantics $\semantics{}$,
they can be strengthened into  $\abstracttracepropertyT$,  $\abstracttraceproperty_{b}$,
and $\abstracttraceproperty_{\!\neg b}$ by intersection with the abstract invariance semantics $\abstractinvariancesemantics{}$.

\begin{example}\label{ex:abstraction0} Continuing Example \ref{ex:abstraction}, an iterated forward reachability/backward accessibility static invariance analyzer can yield the abstract invariance semantics $\abstractinvariancesemantics{}{}$ and strengthened user specifications as follows:\par
%
\bgroup\ttfamily\footnotesize\jot=-35pt\arraycolsep0.75\arraycolsep\begin{eqntabular}{@{}l|c|c|c|c@{}}
\lLstrut&{}^{\ell_1}&{}^{\ell_2}&{}^{\ell_3}&{}^{\ell_4}\nonumber\\\hline
\uLstrut\abstractinvariancesemantics{}{},\ \abstracttracepropertyT&\mathsf{true}&a\in\interval{-1}{1}&a,b\in\interval{-1}{1}&\begin{array}[t]{@{}c@{}}a,b\in\interval{-1}{1},c\in\interval{-2}{2}\end{array}\nonumber\\
\abstracttraceproperty_{b}&\mathsf{true}&a\in\interval{-1}{1}&a=b\in\interval{-1}{1}&c=0\nonumber\\
\lLstrut\abstracttraceproperty_{\!\neg b}&\mathsf{true}&a\in\interval{-1}{1}&a\neq b\in\interval{-1}{1}&c\neq0\renumber{\raisebox{-1ex}[0pt][0pt]{\qed}}\\\hline
\end{eqntabular}
\egroup
\end{example}
Notice that instead of computing the abstract semantics and then intersecting with $\userspec{\tracepropertyT}$, $\userspec{\traceproperty_{b}}$ and $\userspec{\traceproperty_{\!\neg b}}$, it would, in general, be more precise to do four analyses, each of which is for one of these three cases \cite[Section 6]{DBLP:journals/jlp/CousotC92}. 

\subsection{Construction of an Abstract Floyd-Hoare Automaton}

For responsibility analysis, we adopt the more refined abstraction of sets of traces  by Floyd-Hoare automata \cite{DBLP:conf/cav/HeizmannHP13,DBLP:conf/sas/GreitschusDP17}. They are a 
symbolic connected graph/automaton/control flow graph whose nodes/vertices are pairs of a program point ${{}^{\ell}}$ and an abstract invariant $\abstractinvariant({}^{\ell})\in\abstractdomain\setminus\{\bot\}$ 
belonging to the abstract domain 
$\pair{\abstractdomain}{\sqsubseteq}$
and whose arcs/edges correspond to abstract events/labelled basic program actions $a=\action{{}^{\ell}}{{}^{\ell'}}$ from one program point ${}^{\ell}$ to another ${{}^{\ell'}}$ abstracting a concrete event $\triple{s}{a}{s'}\in \tau$ where ${}^{\ell}$ is the control state of $s$ and ${{}^{\ell'}}$ that of $s'$ (the labelling of actions ensure that any action in the graph corresponds to a unique program action $a$). 
An Floyd-Haore automaton represents an over-approximation of all concrete traces, in which states satisfy the abstract invariants when following a path in the graph (the condition $\abstractinvariant({}^{\ell})\neq\bot$ avoids abstract dead paths in the automaton, but there could be abstract paths with no corresponding concrete execution). For
simplicity, program loops are preserved in the graph (but may have to be unrolled for more precision).  


The Floyd-Hoare automaton abstracting $\semantics{}$, $\maximalsemantics{}$, $\traceproperty_{b}$ and $\traceproperty_{\!\neg b}$ (see Fig. \ref{Floyd-Hoare-automaton-for-ex:abstraction} for Example \ref{ex:abstraction}) is built from the program control flow graph,
splitting nodes at labels ${}^{\ell}$ whenever $\abstracttraceproperty_{b}({}^{\ell})\neq\abstracttraceproperty_{\!\neg b}({}^{\ell})$ but still merging at loop heads
to ensure finiteness. The program actions $a$ are added on the arcs/edges $\pair{{}^{\ell}}{{}^{\ell'}}$ for events $\triple{{}^{\ell}}{a}{{}^{\ell'}}$, while abstract invariants $\abstractinvariant({}^{\ell})$ are added on the nodes/vertices ${}^{\ell}$ of the graph. Any  subgraph abstracts a trace property.
From now on, abstract trace properties are represented by a set of paths in the Floyd-Hoare automaton. For example, $\abstracttraceproperty_{b}$ of Example \ref{ex:abstraction0} is the maximal path marked ${}^{\ell_1}$, ${}^{\ell_2}$, ${}^{\ell^a_3}$, ${}^{\ell^a_4}$ in Fig. \ref{Floyd-Hoare-automaton-for-ex:abstraction} .



\subsection{Abstract Observation Analysis}\label{sect:Abstract-observation-analysis}

Given $\maximalsemantics{}$, $\latticeofmaximalproperties$, and $\cognizancefunction{}{}$, the observation analysis consists in computing an abstraction $\abstractobservation{}{}{{}^{\ell}}$ of $\observation{}{}{\maximalsemantics{}, \latticeofmaximalproperties, \cognizancefunction{}{}, \trace{}^{\ell}}$ for all concrete prefix traces $\trace{}^{\ell}$ in the concretization of an abstract trace $\abstracttrace{}^{\ell}$ from the entry node to node labelled ${}^{\ell}$
in the Floyd-Hoare automaton.
We consider the case of the omniscient cognizance $\omniscientcognizance{}{}$ (so that $\observation{}{}{\maximalsemantics{}, \latticeofproperties, \cognizancefunction{}{}, \trace}=\inquiryfunction{}{\maximalsemantics{}, \latticeofproperties,\trace}$) and the behaviors of interest in $\latticeofmaximalproperties$ are all finite. According to the definition of $\inquiryfunction{}{}$ in (\ref{eq:inquiryfunction}), for any $\traceproperty\in\{\traceproperty_{b},\traceproperty_{\!\neg b}\}$ 
which is a set of maximal traces, we need to determine the prefix traces
$\trace\in\prefixesofset{\traceproperty}$ such that $\forall\trace'\in\maximalsemantics{}.\;\trace\isprefix\trace'\Rightarrow\trace'\in\traceproperty$. This can be done in the abstract using the Floyd-Hoare automaton and made more precise using the eventuality static analyzer.


The objective is to mark every node by  $\abstractobservation{}{}{{}^{\ell}}$ which is $\abstracttraceproperty_{b}$, $\abstracttraceproperty_{\!\neg b}$, 
$\abstracttraceproperty_{b}/\abstracttraceproperty_{\!\neg b}$, or $\maximaltracepropertytop$. 
$\abstracttraceproperty_{b}$ (respectively $\abstracttraceproperty_{\!\neg b}$)
means that from ${}^{\ell}$ there exists a reachable concrete state satisfying $\abstractinvariant({}^{\ell})$ from which a concrete execution will definitely reach a terminal 
state satisfying $\abstracttraceproperty_{b}$ (respectively $\abstracttraceproperty_{\!\neg b}$). $\abstracttraceproperty_{b}/\abstracttraceproperty_{\!\neg b}$ means that
there exists a reachable concrete state satisfying $\abstractinvariant({}^{\ell})$ from which there are two concrete executions which will definitely reach a terminal 
state satisfying $\abstracttraceproperty_{b}$ for one and $\abstracttraceproperty_{\!\neg b}$ for the other. $\maximaltracepropertytop$ means that this is a possibility but
not a certitude.

We start with marking every terminal node in the Floyd-Hoare automaton with either $\abstracttraceproperty_{b}$ or $\abstracttraceproperty_{\!\neg b}$ and marking every other node with $\maximaltracepropertybot$.
Then we proceed backwards iteratively, starting form terminal nodes, by considering all nodes ${{}^{\ell}}$ of the Floyd-Hoare automaton. The node ${{}^{\ell}}$ is marked $\abstracttraceproperty_{b}$ (respectively $\abstracttraceproperty_{\!\neg b}$) if all its successors in the Floyd-Hoare automaton are marked $\abstracttraceproperty_{b}$ (respectively $\abstracttraceproperty_{\!\neg b}$) and the eventuality static analyzer can prove that this property holds at ${{}^{\ell}}$ (i.e., as explained above, a terminal node satisfying $\abstracttraceproperty_{b}$ (respectively $\abstracttraceproperty_{\!\neg b}$) is reachable in the concrete). If the node ${{}^{\ell}}$ has one successor marked $\abstracttraceproperty_{b}$
and another one marked $\abstracttraceproperty_{\!\neg b}$, then it is marked $\abstracttraceproperty_{b}/\abstracttraceproperty_{\!\neg b}$ if the eventuality static analyzer can prove that this property holds at ${{}^{\ell}}$ (i.e., as explained above, a terminal node satisfying $\abstracttraceproperty_{b}$ and another terminal node satisfying $\abstracttraceproperty_{\!\neg b}$ are both definitely reachable in the concrete). Otherwise, the node ${{}^{\ell}}$ is marked $\maximaltracepropertytop$. Since the Floyd-Hoare automaton is finite, this fixpoint computation does terminate.

\begin{figure*}[t]
\vspace*{2mm}
	\begin{center}
	\begin{tikzpicture}[scale=0.925,>=stealth',shorten >=1pt,auto,node distance=1.5cm,thick,initial text=]\footnotesize
	\tikzstyle{every state}=[fill=gray!15,draw=black,text=black,scale=0.9]
	
	\node[initial,state
	] at (0,1) (1) {${}^{\ell_1}$\llap{\raisebox{2.3em}[0pt][0pt]{$\top$}}};
	\node[rectangle,draw,thin,inner sep=2pt] at (0,0.15) (m1) {\footnotesize$\abstractobservation{}{}{{}^{\ell_1}}=\traceproperty_{b}/\traceproperty_{\!\neg b}$};
	\node[state
	] at (3,1) (2) {\rlap{\hskip-2em\raisebox{2.3em}[0pt][0pt]{$a\in\interval{-1}{1}$}}${}^{\ell_2}$};
	\node[state] at (6,2) (3) {\rlap{\hskip-3em\raisebox{-2.em}[0pt][0pt]{$a\in\interval{-1}{1}, a=b$}}${}^{\ell^a_3}$};
	\node[rectangle,draw,thin,inner sep=2pt] at (6,-0.9) (m3) {\scriptsize\begin{tabular}{c}$\abstractobservation{}{}{{}^{\ell^b_3}}=\traceproperty_{\!\neg b}$
	\end{tabular}};
	\node[state] at (6,0) (3nb) {\rlap{\hskip-3em\raisebox{2.3em}[0pt][0pt]{$a\in\interval{-1}{1}, a\neq b$}}${}^{\ell^b_3}$};
	\node[state,accepting] at (11,2) (4) {\rlap{\hskip-3em\raisebox{-2.em}[0pt][0pt]{$a\in\interval{-1}{1}, c=0$}}$\color{black}{}^{\ell^a_4}$};
	\node[rectangle,draw,thin,inner sep=2pt] at (11,-0.9) (m4) {\scriptsize\color{gray}\begin{tabular}{c}$\abstractobservation{}{}{{}^{\ell^b_4}}=\traceproperty_{\!\neg b}$
	\end{tabular}};
	\node[state,accepting] at (11,0) (4nb)  {\rlap{\hskip-3em\raisebox{2.3em}[0pt][0pt]{$a\in\interval{-1}{1}, c\neq 0$}}$\color{black}{}^{\ell^b_4}$};
	\node[rectangle,draw,thin,inner sep=2pt] at (11,2.9) (m4nb) {\color{gray}\scriptsize\begin{tabular}{c}$\abstractobservation{}{}{{}^{\ell^a_4}}=\traceproperty_{b}$
	\end{tabular}};
	
	\path (1) edge[->,line width=1.3pt] node {\texttt{\tiny a = input\_1();}} (2)
	      (2) edge[->,line width=1.3pt] node {{}} (3)
	      (2) edge[->] node {\raisebox{0.8em}{\hskip-2.33em\texttt{\tiny b = input\_2();}}} (3nb)
	      (3) edge[->,gray] node {{\texttt{\tiny c = (a - b);}}} (4)
	      (3nb) edge[->,gray] node {{\texttt{\tiny c = (a - b);}}} (4nb)
	          ;
	\node at (-1,2) (b0) {{}};	   
	\node at (6.5,2) (b1) {{}};	   
	\node[rectangle,draw,thin,inner sep=2pt] at (6,2.9) (m4nb) {\scriptsize\begin{tabular}{c}$\abstractobservation{}{}{{}^{\ell^a_3}}=\traceproperty_{b}$
	\end{tabular}};
	\node at (3.5,0.1) (b3) {{}};	   
\node[rectangle,draw,thin,inner sep=2pt] at (3,0.15) (m1) {\footnotesize$\abstractobservation{}{}{{}^{\ell_2}}=\traceproperty_{b}/\traceproperty_{\!\neg b}$};
\end{tikzpicture}%
 \end{center}
 \vskip-.5em
 \caption{Floyd-Hoare automaton for Example \ref{ex:abstraction}, \ref{ex:abstraction0}, \ref{ex:abstraction1}, and \ref{ex:abstraction2}\label{Floyd-Hoare-automaton-for-ex:abstraction}}
\vskip-1em
\end{figure*}

\begin{example}\label{ex:abstraction1}Continuing Example \ref{ex:abstraction0},
the Floyd-Hoare automaton is given in Figure \ref{Floyd-Hoare-automaton-for-ex:abstraction}.
Each node in the automaton is associated with a program point ${}^{\ell}$, and is decorated by an abstract invariant $\abstractinvariant({}^{\ell})$ holding at that point and a boxed abstract observation $\abstractobservation{}{}{{}^{\ell}}$. 



At the initialization all nodes are marked $\maximaltracepropertybot$ 
but the double-circled circled nodes ${}^{\ell^a_4}$ marked $\abstractobservation{}{}{{}^{\ell^a_4}}=\abstracttraceproperty_{b}$
and ${}^{\ell^b_4}$ marked $\abstractobservation{}{}{{}^{\ell^b_4}}=\abstracttraceproperty_{\!\neg b}$. Then
the iteration marks similarly the nodes ${}^{\ell^a_3}$ and ${}^{\ell^b_3}$. The nodes ${}^{\ell_2}$ and ${}^{\ell_1}$ are marked $\abstracttraceproperty_{b}/\abstracttraceproperty_{\!\neg b}$ if the eventuality analyzer provides a guarantee that $\abstracttraceproperty_{b}$ and $\abstracttraceproperty_{\!\neg b}$ will definitely be reachable. Should that eventuality analysis fail, the mark would be $\maximaltracepropertytop$. 
\qed\end{example}

\subsection{Abstract Responsibility Analysis}\label{sect:Abstract-responsibility-analysis}
For an abstract behavior $\abstractbehavior$ of interest, the responsibility analysis searches each of the paths that satisfy the abstraction $\abstracttracepropertyT$ of $\tracepropertyT$ in the Floyd-Hoare automaton forward, from the initial state. This search looks for the first transition from ${}^{\ell}$ to ${}^{\ell'}$ where ${}^{\ell}$ is marked by $\abstractobservation{}{}{{}^{\ell}}=\abstracttraceproperty$ and  ${}^{\ell'}$ is marked $\abstractobservation{}{}{{}^{\ell'}}=\abstracttraceproperty'$, such that $\abstracttraceproperty'\sqsubseteq\abstractbehavior\sqsubsetneq\abstracttraceproperty$. For the sake of simplicity, it is assumed that $\abstractbehavior$ is $\abstracttraceproperty_{b}$ (or alternatively $\abstracttraceproperty_{\!\neg b}$) in Section \ref{sect:User-specification}.

For every elementary path of the Floyd/Hoare automaton from the origin to a terminal node marked $\abstracttraceproperty_{b}$ we look for the first transition $\triple{{}^{\ell}}{a}{{}^{\ell'}}$ from ${}^{\ell}$ to ${}^{\ell'}$ where ${}^{\ell}$ is marked by $\abstractobservation{}{}{{}^{\ell}}=\abstracttraceproperty_{b}/\abstracttraceproperty_{\!\neg b}$ and  ${}^{\ell'}$ is marked $\abstractobservation{}{}{{}^{\ell'}}=\abstracttraceproperty_{b}$, if any. Then the corresponding action $a$ is definitely responsible for $\abstracttraceproperty_{b}$ on all concretizations of that path.

Else, we look on that elementary path of the Floyd/Hoare automaton for the last transition $\triple{{}^{\ell}}{a}{{}^{\ell'}}$ from ${}^{\ell}$ to ${}^{\ell'}$ where ${}^{\ell}$ is marked by $\abstractobservation{}{}{{}^{\ell}}=\maximaltracepropertytop$ and  ${}^{\ell'}$ is marked $\abstractobservation{}{}{{}^{\ell'}}=\abstracttraceproperty_{b}$, if any. Then the corresponding action $a$ is potentially responsible for $\abstracttraceproperty_{b}$ as well as all actions $a'$ before $a$ on that path which concretization have more than one choice (e.g.\ input commands).
(Notice $\triple{{}^{\ell}}{a}{{}^{\ell'}}$ from ${}^{\ell}$ to ${}^{\ell'}$ where ${}^{\ell}$ is marked by $\maximaltracepropertytop$ and ${}^{\ell'}$ by $\abstracttraceproperty_{b}$ is also the first one since all previous ${}^{\ell''}$ before ${{}^{\ell}}$ on that elementary path will be marked $\maximaltracepropertytop$.)

Otherwise, all points ${}^{\ell}$ on that elementary path are marked $\abstractobservation{}{}{{}^{\ell}}=\abstracttraceproperty_{b}$ and no action is responsible for $\abstracttraceproperty_{b}$ on that elementary path of the Floyd/Hoare automaton.

All elementary paths of the Floyd/Hoare automaton must be explored since responsibility may occur in different contexts.

\smallskip
\begin{theorem}\label{theorem:responsibility-analysis-sound}  The abstract responsibility analysis is sound.
\end{theorem}
\begin{proof}
The proof is included in the appendix \ref{subsec:proof-theorem-responsibility-analysis-sound}.
\end{proof}
\smallskip

\begin{example}\label{ex:abstraction2}Continuing Example \ref{ex:abstraction1}, the Floyd-Hoare automaton of Figure \ref{Floyd-Hoare-automaton-for-ex:abstraction}, determines that, for $\abstractbehavior=\abstracttraceproperty_{b}$, on the path ${}^{\ell_1}$, ${}^{\ell_2}$, ${}^{\ell^a_3}$ where $a\in\interval{-1}{1}\wedge a=b$, we have 
$\abstractobservation{}{}{{}^{\ell_2}}=\abstracttraceproperty_{b}/\abstracttraceproperty_{\!\neg b}$ and $\abstractobservation{}{}{{}^{\ell^a_3}}=\abstracttraceproperty_{b}$ so the action \texttt{\small b = input\_2()}  is definitely responsible for the bug (since by
omniscient cognizance the result of \texttt{\small a = input\_1()} is known). In accordance with (\ref{eq:alpharesponsibility}), the future $\tracefuture$ is ignored in the abstraction,
thus the path needs not be explored further (in gray in Figure \ref{Floyd-Hoare-automaton-for-ex:abstraction}).  

If, on the contrary, the eventuality analysis is imprecise and $\abstractobservation{}{}{{}^{\ell_2}}=\maximaltracepropertytop$ (and therefore $\abstractobservation{}{}{{}^{\ell_1}}=\maximaltracepropertytop$) and $\abstractobservation{}{}{{}^{\ell^a_3}}=\abstracttraceproperty_{b}$, the responsibility analysis determine that action \texttt{\small b = input\_2()}  is potentially responsible for the bug and so is the preceding action \texttt{\small a = input\_1()} on that path.
\qed\end{example}

It is important to note that the user specification of $\latticeofmaximalproperties$ is within the abstract domain $\pair{\abstractdomain}{\sqsubseteq}$, so there is no risk that two different properties of interest in $\latticeofmaximalproperties$ have the same abstraction in $\pair{\abstractdomain}{\sqsubseteq}$. For example, if $\abstracttraceproperty_{b}$ and $\abstracttraceproperty_{\!\neg b}$ are abstracted to the same property by a coarser abstraction, the Floyd-Hoare automaton merges ${{}^{\ell^a_3}}$ and ${{}^{\ell^b_3}}$ as well as ${{}^{\ell^a_4}}$ and ${{}^{\ell^b_4}}$ in Fig. \ref{Floyd-Hoare-automaton-for-ex:abstraction}, making the responsibility analysis impossible. 

Another important remark is that the definition $\alpharesponsibility{}{}{}$ of concrete responsibility (\ref{eq:alpharesponsibility}) is definite (i.e. a necessary and sufficient condition). On the contrary, the abstract responsibility proposed in this section may be definite but also potential (i.e. a necessary but maybe not sufficient condition). This ensure that events not responsible in the abstract are definitely not responsible in the concrete. 
Moreover in loops, abstract responsibility does not report at which loop iteration the responsibility of the event is established.  
\subsection{Proofs for Theorem \ref{theorem:responsibility-analysis-sound}\label{subsec:proof-theorem-responsibility-analysis-sound}}
We start with an example showing that if an action is declared potentially responsible then it might not be while one of
the previous action offering a concrete choice might be.
\begin{example}\label{ex:abstraction-top-all-previous}Let us analyze which input is responsible for the bug $\mathtt{c}==0$ at line ${}^{\ell_4}$ of the following program with omniscient cognizance.\par
\bgroup\ttfamily\footnotesize\begin{eqntabular*}[fl]{@{\qquad}L}
${}^{\ell_1}$  a = input\_1(); // Input 1, 0 or -1\\
${}^{\ell_2}$  b = input\_2(); // Input 1 or -1\\
${}^{\ell_3}$  c = (a * b);\\
${}^{\ell_4}$  // c==0 is a bug
\end{eqntabular*}\egroup
The static analysis is with intervals and [dis]equalities but we assume the invariance analyzer cannot determine that $b\neq 0$ when reaching line ${}^{\ell_3}$. Moreover the
eventuality analysis is assumed to be completely ineffective. The Floyd/Hoare
automaton (in absence of eventuality analysis) is given in Figure \ref{Floyd-Hoare-automaton-for-ex:abstraction-top-all-previous}. 
In the concrete, the responsibility definition (\ref{eq:alpharesponsibility}) states that action \texttt{\small a = input\_1();} is responsible
for $\abstracttraceproperty_{b}$ i.e. \texttt{c}=0. The abstract analysis determines that \texttt{\small a = input\_2();} is potentially responsible
for $\abstracttraceproperty_{b}$ which is not false (the path ${{}^{\ell^{b}_2}}$, ${{}^{\ell^{ba}_3}}$, ${{}^{\ell^{ba}_4}}$ is dead but this is not
known in the abstract). However, not adding the previous actions with multi-choices, in this example \texttt{\small a = input\_1();}, would be unsound.
\qed\end{example}

\figureexampleabstractiontopallprevious

\smallskip
\textit{Theorem \ref{theorem:responsibility-analysis-sound}}  The abstract responsibility analysis is sound.
\begin{proof} For simplicity we assume the omniscient congnizance $\omniscientcognizance{}{}$
and the traces of interest $\tracepropertyT$ are all traces of $\tracesemantics$ (otherwise $\tracesemantics$ 
below has to be restricted to the concretization of $\abstracttracepropertyT$).

We assume that the Floyd-Hoare automaton is correct, meaning that every execution trace
in the operational semantics $\tracesemantics$ is in the concretization of a path in the Floyd-Hoare automaton. Since the
operational semantics $\tracesemantics$ is abstracted to the event traces in $\maximalsemantics{}$ in which the events
are chosen to be the transitions $s\transition{a}s'\in\tau$ and states are erased, we have that $\tracesemantics$ 
and $\maximalsemantics{}$ are isomorphic, so the abstraction of $\maximalsemantics{}$ by the Floyd-Hoare automaton
is also sound. However, we do not exclude that the static analysis of Section \ref{sect:static-analyzer} is imprecise
so that the Floyd-Hoare automaton may contain spurious paths.

We also assume that the abstract observation analysis of Section \ref{sect:Abstract-observation-analysis} made
in the lattice $\maximaltracepropertybot$, $\abstracttraceproperty_{b}$, $\abstracttraceproperty_{\!\neg b}$, 
$\abstracttraceproperty_{b}/\abstracttraceproperty_{\!\neg b}$, $\maximaltracepropertytop$ is sound. So
if a node ${}^{\ell}$ is marked $\abstracttraceproperty_{b}$ 
(respectively $\abstracttraceproperty_{\!\neg b}$) in the automaton then there exists a concrete execution in $\tracesemantics/\maximalsemantics{}$ through a state
at ${}^{\ell}$ to a terminal node state satisfying  $\abstracttraceproperty_{b}$ (respectively to $\abstracttraceproperty_{\!\neg b}$).
For a node ${}^{\ell}$  marked $\abstracttraceproperty_{b}/\abstracttraceproperty_{\!\neg b}$, we assume that
the eventuality analysis has determined that there exists a concrete execution in $\tracesemantics/\maximalsemantics{}$ reaching a state
at ${}^{\ell}$ that can be  continued both to a terminal node state satisfying  $\abstracttraceproperty_{b}$ 
and to another terminal node state satisfying $\abstracttraceproperty_{\!\neg b}$. 
This is not guaranteed if the node ${}^{\ell}$ is marked  $\maximaltracepropertytop$.

Consider $\triple{\tracehistory}{\traceresponsible}{\tracefuture}$ in (\ref{eq:alpharesponsibility}) so that the action $a$
of $\traceresponsible={{}^{\ell}}\transition{a}{{}^{\ell'}}$ will be declared responsible for $\abstracttraceproperty_{b}$
in the concrete. We must show that this action $a$ will also be declared responsible in the Floyd-Hoare automaton.

By Lemma \ref{lemma:responsibility-definition}, there exists $\triple{\tracehistory}{\traceresponsible'}{\tracefuture'}$  in
 $\tracesemantics/\maximalsemantics{}$ heading to $\abstracttraceproperty_{\!\neg b}$. By soundness of the Floyd-Hoare
 automaton, there are paths $\triple{\overline{\tracehistory}}{\overline{\traceresponsible}}{\overline{\tracefuture}}$
 to $\abstracttraceproperty_{b}$ and $\triple{\overline{\tracehistory}'}{\overline{\traceresponsible}'}{\overline{\tracefuture}'}$
 to $\abstracttraceproperty_{\!\neg b}$ in the automaton respectively generating, by concretization, $\triple{\tracehistory}{\traceresponsible}{\tracefuture}$
 and $\triple{\tracehistory}{\traceresponsible'}{\tracefuture'}$.

Since the events are transitions $s\transition{a}s'\in\tau$ and abstract events are labelled action $a=\action{{}^{\ell}}{{}^{\ell'}}$
where ${{}^{\ell}}$ and ${{}^{\ell'}}$ are the unique labels in the automaton abstracting $s$ and $s'$ respectively,
and since $\overline{\tracehistory}$ and $\overline{\tracehistory}'$ both abstract $\tracehistory$, we have $\overline{\tracehistory}=\overline{\tracehistory}'$.
There are two cases.

\hyphen{4}\quad If all nodes on ${\overline{\tracefuture}}$ are marked  $\abstracttraceproperty_{b}$ or  $\abstracttraceproperty_{b}/\abstracttraceproperty_{\!\neg b}$
then, by Section \ref{sect:Abstract-responsibility-analysis}, the action on $\overline{\traceresponsible}$ which is the same as that on $\traceresponsible$ is declared
definitely responsible, as required. 

\hyphen{4}\quad Otherwise, if some node on ${\overline{\tracefuture}}$ is marked $\maximaltracepropertytop$ then, since the end of the path is marked 
$\abstracttraceproperty_{b}$ there is some later action on ${\overline{\tracefuture}}$ in the Floyd-Hoare automaton which is marked potentially responsible. But then
all actions on that path ${\overline{\tracefuture}}$ offering a choice in the concrete may also be responsible so action $a$ will be declared potentially responsible.
\end{proof}

\fi

\end{document}